\documentclass[a4paper,onecolumn,10pt,accepted=2024-08-01]{quantumarticle}
\pdfoutput=1
\usepackage[utf8]{inputenc}
\usepackage[english]{babel}
\usepackage[T1]{fontenc}

\usepackage[affil-it]{authblk}
\usepackage[usenames,dvipsnames]{xcolor}
\usepackage{amsfonts}
\usepackage{amsmath,amsthm,amssymb}
\usepackage{enumerate}
\usepackage[english]{babel}
\usepackage{graphicx}	
\usepackage[caption=false]{subfig}
\usepackage{url}
\usepackage{bbm}
\usepackage{blkarray}

\usepackage{tikz,pgfplotstable}
\usetikzlibrary{chains}
\usetikzlibrary{fit}
\usetikzlibrary{spy}
\usepackage{pgflibraryarrows}		
\usepackage{pgflibrarysnakes}		

\usepackage{makecell}

\usepackage{epsfig}
\usetikzlibrary{shapes.symbols,patterns} 
\usepackage{pgfplots}

\usepackage{mathrsfs}

\usepackage{hyperref}
\hypersetup{colorlinks=true,citecolor=blue,linkcolor=blue,filecolor=blue,urlcolor=blue,breaklinks=true}

\usepackage{nicefrac}
\usepackage{mathtools}
\usepackage{xparse}
\usepackage{subfig}
\usepackage{parskip}

\theoremstyle{plain}
\newtheorem{theorem}{Theorem}[section]
\newtheorem{lemma}[theorem]{Lemma}

\newtheorem{proposition}[theorem]{Proposition}

\theoremstyle{definition}
\newtheorem{definition}[theorem]{Definition}
\newtheorem{remark}[theorem]{Remark}

\newcommand*{\cA}{\mathcal{A}}
\newcommand*{\cB}{\mathcal{B}}
\newcommand*{\cC}{\mathcal{C}}

\newcommand*{\cH}{\mathcal{H}}
\newcommand*{\cL}{\mathcal{L}}

\newcommand*{\cR}{\mathcal{R}}

\newcommand*{\cW}{\mathcal{W}}
\newcommand*{\cX}{\mathcal{X}}
\newcommand*{\cY}{\mathcal{Y}}

\newcommand*{\RR}{\mathbb{R}}

\newcommand*{\CC}{\mathbb{C}}

\newcommand*{\NN}{\mathbb{N}}

\newcommand*{\id}{I}

\newcommand*{\ket}[1]{| #1 \rangle}
\newcommand*{\bra}[1]{\langle #1 |}

\newcommand{\proj}[1]{|#1\rangle\!\langle #1|}

\newcommand*{\<}{\langle}
\renewcommand*{\>}{\rangle}

\usepackage[normalem]{ulem}

\newcommand{\imH}[1]{H_{(#1)}} 

\newcommand{\imHup}[1]{\imH{#1}^\uparrow}

\newcommand{\tr}[1]{\mathrm{Tr}\left[#1\right]} 
\newcommand{\ptr}[2]{\mathrm{Tr}_{#1}\left[#2\right]}

\let\inner\relax
\NewDocumentCommand\inner{mg}{%
	\ensuremath{\left\langle #1| \IfNoValueTF{#2}{#1}{#2}\right\rangle}%
}
\let\outer\relax
\NewDocumentCommand\outer{mg}{%
	\ensuremath{\ket{#1}\! \IfNoValueTF{#2}{\bra{#1}}{\bra{#2}}}%
}

\begin{document}

	\title{{Device-independent lower bounds on the conditional von Neumann entropy}}
	\author[1,3]{Peter Brown}
	\email{peter.brown@telecom-paris.fr}
	\orcid{0000-0001-9593-0136}
	\author[2]{Hamza Fawzi}
	\author[1]{Omar Fawzi}
	
	\affil[1]{\small{Univ Lyon, ENS Lyon, UCBL, CNRS,  Inria, LIP, F-69342, Lyon Cedex 07, France}}
	\affil[2]{\small{DAMTP, University of Cambridge, United Kingdom}}
	\affil[3]{\small{Télécom Paris, LTCI, Institut Polytechnique de Paris, 91120 Palaiseau, France}}

	\maketitle

\begin{abstract}
	The rates of several device-independent (DI) protocols, including quantum key-distribution (QKD) and randomness expansion (RE), can be computed via an optimization of the conditional von Neumann entropy over a particular class of quantum states. In this work we introduce a numerical method to compute lower bounds on such rates. We derive a sequence of optimization problems that converge to the conditional von Neumann entropy of systems defined on general separable Hilbert spaces. Using the Navascu\'es-Pironio-Ac\'in hierarchy we can then relax these problems to semidefinite programs, giving a computationally tractable method to compute lower bounds on the rates of DI protocols. Applying our method to compute the rates of DI-RE and DI-QKD protocols we find substantial improvements over all previous numerical techniques, demonstrating significantly higher rates for both DI-RE and DI-QKD. In particular, for DI-QKD we show a minimal detection efficiency threshold which is within the realm of current capabilities. Moreover, we demonstrate that our method is capable of converging rapidly by recovering all known tight analytical bounds up to several decimal places. Finally, we note that our method is compatible with the entropy accumulation theorem and can thus be used to compute rates of finite round protocols and subsequently prove their security. 
\end{abstract}

\section{Introduction}
Quantum cryptography enables certain cryptographic tasks to be performed securely with their security guaranteed by the physical laws of nature as opposed to assumptions of computational hardness that are used in more conventional cryptography. However, standard quantum cryptography protocols, for instance BB84~\cite{BB84}, still require a strong level of trust in the hardware used and its implementation. Faulty hardware or malicious attacks on the implementation can compromise the security of these protocols, rendering them useless~\cite{LWWESM}. Whilst better security analyses and improved hardware checks can make it more difficult for an adversarial party to eavesdrop on the protocol, it is always possible that there are unknown side-channel attacks that remain exploitable.

Fortunately, quantum theory also offers a way to remove strong assumptions on the hardware and implementation by instead running a so-called device-independent protocol. Device-independent protocols offer the pinnacle of security guarantees. By relying on minimal assumptions, they remain secure even when the devices used within the protocol are completely untrusted. The central idea behind many device-independent protocols, including randomness expansion (RE) and quantum key distribution (QKD), is that there are certain correlations between multiple separate systems that ($i$) could only have been produced by entangled quantum systems and $(ii)$ are intrinsically random. Intuitively, these \emph{nonlocal correlations} then act as a certificate that guarantees the systems produced randomness~\cite{MAG}. Furthermore, nonlocal correlations have additional applications such as certifying the dimension of a system~\cite{BPAGMS08}, reducing communication complexity~\cite{BCMW10} and in certain cases nonlocal correlations can even certify the exact systems used (up to some unavoidable symmetries)~\cite{SB20}. 

Beginning with the works of~\cite{MayersYao, Ekert}, significant effort has been placed into developing new device-independent protocols and subsequently proving their security. For tasks like DI-RE and DI-QKD, security is now well understood, with tools like the entropy accumulation theorem (EAT)~\cite{DFR,DF} and quantum probability estimation (QPE)~\cite{QPE} enabling relatively simple proofs of security against all-powerful quantum adversaries~\cite{ADFRV,ARV}. In both cases a security proof can be readily established once one has developed a quantitative relationship between the nonlocal correlations observed in the protocol and the quantity of randomness generated by the devices used. In particular, one is required to understand the \emph{rate} of a protocol, i.e. the number of random bits or key bits that are generated per round. The asymptotic rates of both DI-RE and DI-QKD protocols are quantified in terms of the conditional von Neumann entropy. Moreover, using the EAT, one can also bound the total randomness produced in a finite round protocol in terms of the conditional von Neumann entropy.\footnote{For the QPE one must compute a related quantity, namely construct a quantum estimation factor but we leave such investigations to future work.}

A central problem in device-independent cryptography then remains of how one actually computes this conditional von Neumann entropy? More specifically, one is required to compute (or lower bound) the minimum conditional von Neumann entropy produced by the devices in a single round of the protocol conditioned on the adversary's side information. This computation should be device-independent in the sense that there are no restrictions on what systems were used except for some constraints on the expected correlations produced by the parties' devices, e.g., a Bell-inequality violation. For the purpose of the introduction, we use the following informal notation $\inf_{\rho} H[\rho]$ for this quantity of interest, where the infimum is understood to be over all states that are compatible with the observed correlations, and $H$ refers to the conditional von Neumann entropy. We refer to Section~\ref{sec:DI} for a more formal definition. At first glance this optimization would appear very challenging: the optimization itself is non-convex and furthermore we place no restrictions on the dimensions of the quantum systems used. Nevertheless, in spite of these apparent difficulties, significant progress has still been made towards solving such an optimization problem. For example, when the devices in the protocol have only binary inputs and outputs, it is often possible to exploit Jordan's lemma~\cite{Jordan} in order to restrict the optimization to qubit systems. Using this technique, the work of~\cite{PABGMS} was able to give an exact analytical solution to the optimization problem for the one-sided conditional von Neumann entropy when the devices were constrained by a violation of the CHSH inequality~\cite{CHSH}. More recently, other works have also managed to obtain analytical lower bounds on devices satisfying MABK inequality violations~\cite{GMKB21}, tight bounds bounds for the Holz inequality~\cite{GMKB21b} and tight bounds for the asymmetric CHSH inequality~\cite{WAP20}. These analytical results are difficult to obtain and also rely on the reduction to qubit systems, meaning that it is not possible to perform such an analysis when the inputs and outputs of the devices are not binary. As such, general numerical techniques were also developed that could handle more complicated protocols and general linear constraints on the expected correlations. 

For example, in~\cite{BSS14, NPS14} it was shown that the analogous optimization for the min-entropy~\cite{KRS} can be straightforwardly lower bounded using the Navascu\'es-Pironio-Ac\'in (NPA) hierarchy~\cite{NPA,NPA_general}. As the min-entropy is never larger than the von Neumann entropy the results of these computations give lower bounds on the rates. However, one major drawback of this method is that the resulting lower bounds are in general quite loose. More recently, two other proposals for general numerical approaches have been developed. In~\cite{TSGPL19}, the authors give an explicit method to lower bound the conditional von Neumann entropy via a non-commutative polynomial optimization problem that can then be lower bounded using the NPA hierarchy. In~\cite{BFF21}, the authors introduced new entropies that are all lower bounds on the conditional von Neumann entropy and which can, like the min-entropy, be optimized using the NPA hierarchy. Both works improve upon the min-entropy method but neither has been shown to give tight bounds on the actual rate of a protocol and in general there appears to be significant room for improvement. As such, the question remains as to whether one can give a computationally tractable method to compute tight lower bounds on the rates of protocols.

\subsection{Contributions of this work}
In this work, we introduce a new method to solve this problem. To achieve this, we develop a family of variational expressions indexed by a positive integer $m$ that approximate the conditional von Neumann entropy to arbitrary accuracy as $m$ grows. For any $m$, the variational expression is determined by a noncommutative polynomial $P_m$ and for a state given by the density operator $\rho$ takes the form
\begin{align}
\label{eq:var_expr_intro}
H_m[\rho] := \inf_{Z_1, \dots, Z_n} \: \tr{\rho P_m(Z_1, \dots, Z_n, Z_1^*, \dots, Z_n^*)}  \ ,
\end{align}
where $Z_1, \dots, Z_n$ are bounded operators. For any choice of $m$, this expression is a lower bound on the von Neumann entropy i.e., $H[\rho] \geq H_m[\rho]$, and in particular can be used to lower bound the rates of DI protocols. Such variational expressions are closely related to Kosaki's expressions for the quantum relative entropy~\cite{Kos86,Don86,OP04}. However, Kosaki's variational expression has a number of variables ($n$ in~\eqref{eq:var_expr_intro}) that is infinite and thus it is not suitable for computations. The variational expressions we obtain can be seen as well-chosen finite approximations of Kosaki's formula. Note that we need to be careful in the choice of approximation as we need to obtain expressions which always give lower bounds on the conditional von Neumann entropy.

The point of obtaining expressions of the form~\eqref{eq:var_expr_intro} is that when further taking the infimum over all states $\rho$ that are compatible with the observed correlations, it is possible to relax such an optimization problem to a semidefinite program (SDP) using the NPA hierarchy to yield a computationally tractable lower bound. Moreover, we show that for any fixed $m$, the resulting noncommutative polynomial optimization problem can be made to satisfy the property that all the variables have a fixed upper bound on their operator norm. As a consequence, we obtain that the sequence of SDPs given by the NPA hierarchy converges in the limit to $\inf_{\rho} H_m[\rho]$ where $\rho$ is compatible with the observed statistics. 
We also remark that unlike previous numerical techniques that were developed, our analysis applies when the systems are also infinite dimensional.

We then apply this to compute the rates of both DI-RE and DI-QKD protocols. Compared to the other numerical techniques we show significant improvements on the calculated rates. We give improved bounds on the DI randomness certifiable with qubit systems which could be used to yield more efficient experiments for DI randomness~\cite{diexperiment1,diexperiment2}. In addition we also give new lower bounds on the minimal detection efficiency required to perform DI-QKD with qubit systems. This gives a promising approach to conduct DI-QKD experiments with current technologies. We also find that in practice our method can converge quickly as we demonstrate that we can recover (up to several decimal places) tight analytical bounds from~\cite{PABGMS},~\cite{GMKB21b} and~\cite{WAP20}. We also remark that the computations we ran in several of our examples were also vastly more efficient than the two numerical approaches of~\cite{TSGPL19, BFF21}. Finally, we also explain how our technique can be used directly with the entropy accumulation theorem in order to compute non-asymptotic rates of protocols and subsequently prove their security.

The article is structured as follows. In Section~\ref{sec:results} we begin by stating the main results relevant to device-independent cryptography and their application. Note that in this section we will restrict to a special case of our more general result that is sufficient for those interested solely in applications. Later in Section~\ref{sec:methods} we will state the general technical result along with its proof. 

\section{Results}\label{sec:results}

Before we begin let us establish some notation. Let $H$ be a separable Hilbert space, we denote the set of bounded operators on $H$ to itself by $B(H)$. A \emph{state} on $H$ is a positive semidefinite, trace-class operator $\rho$ such that $\tr{\rho} = 1$. We denote the set of states on $H$ by $D(H)$. Given two positive semidefinite operators $\rho$, $\sigma$ on $H$ we write $\rho \ll \sigma$ if $\ker{\sigma}  \subseteq \ker{\rho}$ where $\ker X := \{\ket{v} \in H : X \ket{v} = 0\}$. Throughout this work $\RR_+$ will denote the set of nonnegative real numbers, $\NN$ will denote the set of strictly positive integers, $\ln$ will denote the natural logarithm and $\log_2$ will denote the logarithm base two. In this section, for simplicity of presentation, we assume that $H$ is finite-dimensional. See Section~\ref{sec:methods} for the setting where $H$ is infinite-dimensional.

\subsection{Converging upper bounds on the relative entropy}

We define the relative entropy between two positive semidefinite operators $\rho$ and $\sigma$ as 
\begin{equation}
	D(\rho\|\sigma) := \tr{\rho \left(\log_2 \rho - \log_2 \sigma\right)}
\end{equation}
whenever $\rho \ll \sigma$ and $+\infty$ otherwise. 
Note that this is equivalent to the definition 
\begin{equation}\label{eq:Ddef2}
	D(\rho\|\sigma) = \sum_{i,j} y_i \log_2(y_i / x_j) |\!\inner{\psi_i}{\phi_j}\!|^2,
\end{equation}
where $\{\ket{\psi_i}, y_i\}_i$ are an orthonormal basis of eigenvectors with their corresponding eigenvalues for the operator $\rho$ and similarly $\{\ket{\phi_j}, x_j\}_j$ for $\sigma$. 
 For a bipartite state $\rho_{AB}$ on a  Hilbert space $H_A \otimes H_B$ we define the conditional von Neumann entropy as 
\begin{equation}
	H(A|B) := - D(\rho_{AB} \| \id_A \otimes \rho_B)\,.
\end{equation}

The main technical result of this work is the following theorem. Note that a more general version of this theorem, stated for von Neumann algebras is given in Theorem~\ref{thm:D_upper_bounds}.
\begin{theorem}\label{thm:D_upper_bounds_simp}
	Let $H$ be a finite-dimensional Hilbert space, $\rho, \sigma$ be two positive semidefinite operators on $H$ 
	and assume $\lambda > 0$ is such that $\rho \leq \lambda \sigma$.
	Then for any $m \in \NN$ there exists a choice of $t_1,\dots t_m \in (0,1]$ and $w_1,\dots,w_m > 0$, such that
\begin{equation}\label{eq:D_ub_boundedops}
	\begin{aligned}
	D(\rho\|\sigma) \leq - c_m\!-\!\sum_{i=1}^{m-1} \frac{w_i}{t_i \ln 2}\,\inf_{Z}& \, \tr{\rho (Z+Z^*+ (1-t_i)Z^* Z} + t_i \tr{\sigma Z Z^*} \ \\
	\mathrm{s.t.}& \quad  \|Z\|\leq \frac{3}{2} \max\left\{\frac{1}{t_i}, \frac{\lambda}{1-t_i}\right\}
	\end{aligned}
\end{equation}
where $c_m = \tr{\rho}(\sum_{i=1}^m \frac{w_i}{t_i \ln 2} -\frac{\lambda}{m^2 \ln 2})$. 
Moreover, the RHS converges to $D(\rho \|\sigma)$ as $m \to \infty$.
\end{theorem}

\begin{remark}
	The constants $t_i$ and $w_i$ appearing in Theorem~\ref{thm:D_upper_bounds_simp} are the nodes and weights of a Gauss-Radau quadrature rule over $[0,1]$ with endpoint $t_m = 1$. Importantly they are efficient to compute and we refer the reader to Sec.~\ref{sec:approx_log} for further details as well as~\cite{oisinpaper} for a more general treatment.
\end{remark}

The above theorem provides a convergent sequence of upper bounds on the relative entropy in the form of an optimization problem. This optimization has several features that are crucial to the applications we will now see. In particular it has an objective function that is linear in the operators $\rho$ and $\sigma$ and the form of the optimization does not change with the dimension. In the following subsection we will show how to turn these upper bounds on the relative entropy into semidefinite programming lower bounds on the rates of DI protocols. 

\subsection{Application to device-independent cryptography}
\label{sec:DI}

For the device-independent setup we consider two honest parties\footnote{Our work extends straightforwardly to more parties, e.g., see the multipartite results in Figure~\ref{fig:holz_comparison}. However, for brevity and clarity we will focus on two parties in our exposition.}, Alice and Bob, and an adversary Eve. Alice and Bob each have access to a black-box device which they can give inputs to and receive outputs from. We shall denote the inputs of Alice and Bob by $X$ and $Y$ respectively. Similarly we denote the outputs of Alice and Bob by $A$ and $B$ respectively. All inputs and outputs come from finite sets. We refer to a setup wherein Alice's device has $n_A$ inputs and $m_A$ outputs and Bob's device has $n_B$ inputs and $m_B$ outputs as a \emph{$n_An_Bm_Am_B$-scenario}. A \emph{round} consists of Alice and Bob performing the following: (1) they shield\footnote{This shielding could be enforced via spacelike separation for example.} their devices so that they cannot communicate with one another; (2) Alice and Bob each provide their respective device with an input, selected using some fixed probability distribution $p(x,y)$; (3) they each receive an output from their respective device. In a device-independent protocol such as DI-QKD or DI-RE the statistics of many rounds will be analyzed in order to determine whether or not randomness is being produced from the devices or whether a secret key can be distilled. 

We assume that the devices are constrained by quantum theory and that they act in the following way. Let $Q_A$, $Q_B$ and $Q_E$ be the three separable Hilbert spaces of Alice's device, Bob's device and Eve's device respectively. At the beginning of each round a tripartite state $\rho_{Q_AQ_BQ_E}$ on $Q_A \otimes Q_B \otimes Q_E$ is shared between the three systems. In response to an input, Alice and Bob's devices will measure some preselected POVMs $\{\{M_{a|x}\}_a\}_x$, $\{\{N_{b|y}\}_b\}_y$ on the parts of the state that they received and return the measurement outcome. The state and measurements are unknown to Alice and Bob but may be known by Eve. Overall, the joint conditional distribution of their outputs may be described via the Born rule as
\begin{equation}
p(a,b|x,y) = \tr{\rho_{Q_AQ_BQ_E} (M_{a|x} \otimes N_{b|y} \otimes I)}.
\end{equation}
Immediately after the round, the joint state between the classical information recorded by the honest parties and Eve's quantum system may be described by the cq-state
\begin{equation}\label{eq:pms}
	\rho_{ABXYQ_E} = \sum_{abxy} p(x,y) \outer{abxy} \otimes \rho_{Q_E}^{abxy},
\end{equation}
where 
\begin{equation}
\rho_{Q_E}^{abxy} = \ptr{Q_AQ_B}{\rho_{Q_AQ_BQ_E} (M_{a|x} \otimes N_{b|y} \otimes I)}\,.
\end{equation}
In spot-checking protocols, secret key and randomness are only extracted from rounds with a particular input~\cite{MS2}. As such, we henceforth consider some distinguished inputs $(x^*,y^*)$ and the post-measurement state $\rho_{ABQ_E} = \sum_{ab} \outer{ab} \otimes \rho_{Q_E}^{abx^*y^*}$. The exact choices of $(x^*, y^*)$ will be made explicit when we compute rates for given protocols. For a spot-checking DI-RE protocol, the asymptotic rate is characterized by 
\begin{equation}
	H(AB|X=x^*, Y=y^*, Q_E)
\end{equation}
if the randomness is extracted from both Alice and Bob's outputs and
\begin{equation}
	H(A|X=x^*, Q_E)
\end{equation}
if the randomness is only extracted from the outputs of Alice's device. 
For a spot-checking DI-QKD protocol with one-way error correction, the asymptotic rate is given in terms of the Devetak-Winter bound~\cite{DW05}
\begin{equation}
	H(A|X=x^*, Q_E) - H(A|B,X=x^*,Y=y^*).
\end{equation}

To obtain the asymptotic\footnote{For protocols with a finite number of rounds a more careful analysis of the rate is required. However general techniques like the entropy accumulation theorem~\cite{DFR,DF} allow us to lift the asymptotic rate to the rate of a protocol with a finite number of rounds.} rate of the device-independent protocols we must consider a minimization of the above quantities over all possible quantum systems that could have produced the statistics we observed on expectation. More formally, let $\cA, \cB, \cX, \cY$ denote the finite sets from which the random variables $A,B,X,Y$ take their values. Then given a finite set $\cC$ let $C : \cA\cB\cX\cY \rightarrow \cC$ be some function which will act as a statistical test on $ABXY$. Finally, we consider a probability distribution $q(c)$ on $\cC$. We then say a distribution $p(a,b,x,y)$ on $ABXY$  is \emph{compatible} with the pair $(C, q)$ if
\begin{equation}
	\sum_{(a,b,x,y): \,C(a,b,x,y) = c} p(a,b,x,y) = q(c).
\end{equation}
In other words, when we apply our statistical test $C$ to the random variables $ABXY$ we obtain a new random variable whose distribution is $q$.
For example, let $\cA = \cB = \cX = \cY =\cC = \{0,1\}$,  
\begin{equation}
C(a,b,x,y) = \begin{cases}
	1 \qquad \text{if } a \oplus b = x y  \\
	0 \qquad \text{otherwise }
\end{cases}
\end{equation}
and let $q(1) = \omega$. Then for $p(x,y) = 1/4$ the pair $(C,q)$ imposes the constraint that the distribution $p(a,b,x,y)$ should achieve an expected score of $\omega$ in the CHSH game.

More generally we say that a tuple $(Q_AQ_BQ_E, \rho, \{M_{a|x}\}, \{N_{b|y}\})$ is \emph{compatible} with the constraints $(C,q)$ if the probability distribution $p(a,b,x,y) = p(x,y)\tr{\rho (M_{a|x} \otimes N_{b|y} \otimes \id)}$ is compatible with $(C,q)$. We refer to such a tuple as a \emph{strategy}. For each strategy we can also associate a post-measurement state via~\eqref{eq:pms}. Then for a fixed $(C,q)$ the \emph{device-independent} rate of randomness produced by Alice's device on input $X=x^*$, if the devices satisfy the constraints imposed by $(C,q)$, is given by 
\begin{equation}\label{eq:rate_optimization}
	\inf H(A|X=x^*, Q_E)
\end{equation}
where the infimum is taken over all strategies compatible with $(C,q)$ and the conditional von Neumann entropy is evaluated on the post-measurement state induced by the strategy. The rates for the randomness produced by both devices and the device-independent Devetak-Winter rate can be defined analogously by replacing the objective function with the appropriate quantity. Note that by considering appropriate dilations we can restrict the optimization to strategies wherein the measurements are projective and the state $\rho_{Q_AQ_BQ_E}$ is pure. 

The following lemma demonstrates how to use the upper bounds on the relative entropy from Theorem~\ref{thm:D_upper_bounds_simp} in order to lower bound this infimum by a converging sequence of optimizations that can be subsequently lower bounded using the NPA hierarchy.
\begin{lemma}\label{lem:di_rewriting}
	Let $m \in \NN$ and let $t_1,\dots,t_m$ and $w_1,\dots,w_m$ be the nodes and weights of an $m$-point Gauss-Radau quadrature on $[0,1]$ with an endpoint $t_m = 1$, as specified in Theorem~\ref{thm:gauss-radau-quad}. Let $\rho_{Q_AQ_BQ_E}$ be the initial quantum state shared between the devices of Alice, Bob and Eve and let $\{M_{a|x^*}\}$ denote the measurements operators performed by Alice's device in response to the input $X=x^*$. Furthermore for $i = 1, \dots, m-1$ let $\alpha_i = \tfrac32\max\{\tfrac{1}{t_i}, \tfrac{1}{1-t_i}\}$. Then $H(A|X=x^*, Q_E)$ is never smaller than
	\begin{equation}\label{eq:H_lower_bound}
		\begin{aligned}
			c_m + \sum_{i=1}^{m-1} \frac{w_i}{t_i \ln 2} \sum_a \inf_{Z_a \in B(Q_E)}&  \tr{\rho_{Q_AQ_E}\bigg(M_{a|x^*} \otimes (Z_a + Z_a^* + (1-t_i) Z_a^* Z_a) + t_i (\id_{Q_A} \otimes Z_a Z_a^*)\bigg)} \\
			\mathrm{s.t.}& \quad \|Z_a\| \leq \alpha_i
		\end{aligned}
	\end{equation}
	where $c_m = \sum_{i=1}^{m-1} \frac{w_i}{t_i \ln 2}$. Moreover these lower bounds converge to $H(A|X=x^*, Q_E)$ as $m \to \infty$.
	\begin{proof}
		Let $\rho_{AQ_E} = \sum_{a} \outer{a} \otimes \rho_{Q_E}(a,x^*)$ be the cq-state after Alice has performed her measurement corresponding to the input $x^*$, i.e., $\rho_{Q_E}(a,x^*) = \ptr{Q_AQ_B}{\rho_{Q_AQ_BQ_E} (M_{a|x^*} \otimes \id \otimes \id)}$. Then for any $m \in \NN$, using the relation $H(A|X=x^*, Q_E) = - D(\rho_{AQ_E}\|\id_A\otimes \rho_{Q_E})$ and Theorem~\ref{thm:D_upper_bounds_simp} we have that $H(A|X=x^*, Q_E)$ is never smaller than
		\begin{equation}
\begin{aligned}
			 c_m + \sum_{i=1}^{m-1} \frac{w_i}{t_i \ln 2}\,&\inf_{Z \in B(H)}  \tr{\rho_{AQ_E} (Z+Z^* + (1-t_i) Z^* Z)} + t_i \tr{(\id_A \otimes\rho_{Q_E}) Z Z^*}  \\
			& \quad\mathrm{s.t.} \quad \| Z \| \leq \frac{3}{2} \max \left\{\frac{1}{t_i}, \frac{\lambda}{1-t_i}\right\}	
\end{aligned}	
	\end{equation}
		where $c_m = - \frac{\lambda}{m^2 \ln 2} + \sum_{i=1}^m \frac{w_i}{t_i \ln 2}$ and $\lambda$ is some real number such that $\rho_{AQ_E} \leq \lambda \id_A \otimes \rho_{Q_E}$. Now as system $A$ is finite dimensional, we can write the operator $Z = \sum_{ab} \outer{a}{b} \otimes Z_{(a,b)}$ for some operators $Z_{(a,b)} \in B(Q_E)$. Then for the first term we have
		\begin{equation}
		\begin{aligned}
			\tr{\rho_{AQ_E} (Z + Z^*)} &= \sum_a \tr{\rho_{Q_E}(a,x^*)(Z_{(a,a)} + Z_{(a,a)}^*)} \\
			&= \sum_a \tr{\ptr{Q_A}{\rho_{Q_AQ_E} (M_{a|x^*} \otimes \id_{Q_E})}(Z_{(a,a)} + Z_{(a,a)}^*)} \\
			&= \sum_a \tr{\rho_{Q_AQ_E} (M_{a|x^*} \otimes (Z_{(a,a)} + Z_{(a,a)}^*))}			
		\end{aligned}
		\end{equation}
		where on the first line we traced out the $A$ system, on the second line we substituted in the definition of $\rho_{AQ_E}(a,x^*)$ and on the third line we used the identity $\ptr{A}{X_{AB} (\id \otimes Y_B)} = \ptr{A}{X_{AB}}Y_B$. Repeating this for the second term we find
		\begin{equation}
		\begin{aligned}
			\tr{\rho_{AQ_E} Z^*Z} &= \sum_{ab} \tr{\rho_{Q_E}(a,x^*) Z^*_{(b,a)} Z_{(b,a)}} \\
			&\geq \sum_a  \tr{\rho_{Q_E}(a,x^*) Z^*_{(a,a)} Z_{(a,a)}} \\
			&= \sum_a \tr{\rho_{Q_AQ_E} (M_{a|x^*} \otimes Z_{(a,a)}^*Z_{(a,a)})}
		\end{aligned}
		\end{equation}
		where on the second line we noted that $\sum_b Z_{(b,a)}^* Z_{(b,a)} \geq Z_{(a,a)}^* Z_{(a,a)}$. Finally we get a similar relation for the final term
		\begin{equation}
		\begin{aligned}
			\tr{(\id_A \otimes \rho_{Q_E}) Z Z^*} \geq \sum_{a} \tr{\rho_{Q_AQ_E} (\id_{Q_A} \otimes Z_{(a,a)} Z^*_{(a,a)})}.
		\end{aligned}
		\end{equation}
		Inserting these three rewritings into the lower bound on $H(A|X=x^*, Q_E)$ and relabelling $Z_{(a,a)}$ to $Z_a$ we recover the objective function stated in the lemma. Note that the above rewritings and the fact that we are minimizing implies that we need only consider operators $Z$ that are block diagonal in the sense that $Z = \sum_a \outer{a} \otimes Z_{(a,a)}$.
		
		As $\rho_{AQ_E}$ is a cq-state we have $\rho_{AQ_E} \leq \id_A \otimes \rho_{Q_E}$ and we can set $\lambda = 1$. This recovers the form of $c_m$ stated in the lemma (noting that $w_m = \frac{1}{m^2}$ and $t_m=1$). As it is sufficient to take $Z= \sum_a \outer{a} \otimes Z_a$, we must have $\alpha_i = \frac{3}{2} \max \left\{\frac{1}{t_i}, \frac{\lambda}{1-t_i}\right\} \geq \|\sum_a \outer{a} \otimes Z_a\| = \max_a \|Z_a\|$. 
		
		Finally the convergence statement follows immediately from the convergence proven in Theorem~\ref{thm:D_upper_bounds}.
	\end{proof}
\end{lemma}

\begin{remark}[Adapting to other entropies]
	The above lemma only describes a bound for $H(A|X= x^*,Q_E)$. However the proof can be easily adapted to the case of the global entropy $H(AB|X=x^*, Y=y^*, Q_E)$ or for non-fixed inputs, e.g., $H(A|XQ_E)$. For example, the global entropy $H(AB|X=x^*, Y=y^*, Q_E)$ can be lower bounded by replacing the inner summation in~\eqref{eq:H_lower_bound} with 
	\begin{equation}
		\sum_{ab} \inf_{Z_{ab} \in B(Q_E)} \tr{\rho_{Q_AQ_BQ_E}\bigg(M_{a|x^*}  \otimes N_{b|y^*} \otimes (Z_{ab} + Z_{ab}^* + (1-t_i) Z_{ab}^* Z_{ab}) + t_i (\id_{Q_AQ_B} \otimes Z_{ab} Z_{ab}^*)\bigg)}.
	\end{equation}
	Similarly one could also adapt the proof to bound $H(A|XQ_E)$, allowing the entropy to averaged over the inputs $X$. For this one would replace the inner summation with
	\begin{equation}
		\sum_{ax} \inf_{Z_{ax} \in B(Q_E)} p(x) \tr{\rho_{Q_AQ_E} \bigg( M_{a|x} \otimes (Z_{ax} + Z_{ax}^* + (1-t_i) Z_{ax}^* Z_{ax})  + t_i (\id_{Q_A} \otimes Z_{ax}Z_{ax}^*) \bigg)}.
	\end{equation}
\end{remark}

The above lemma and remark provide a converging sequence of lower bounds on the conditional von Neumann entropy. In order to turn these into lower bounds on the rate of a device-independent protocol we must also include the optimizations of all states, measurements and Hilbert spaces, subject to any constraints on the devices' joint probability distribution that we wish to impose. Suppose for $1 \leq j \leq r$ for some $r \in \NN$ we impose on Alice and Bob's devices a collection of constraints 
\begin{equation}
\sum_{abxy} c_{abxyj}p(a,b|x,y) \geq v_j 
\end{equation}
where $c_{abxyj}, v_j \in \RR$. Then using Lemma~\ref{lem:di_rewriting} we can compute a lower bound on $H(A|X=x^*, Q_E)$ for all possible devices that satisfy the above constraints by solving the following optimization problem:
\begin{equation}\label{eq:full_di_optimization_tensor}
	\begin{aligned}
		c_m + &\inf \quad  \sum_{i=1}^{m-1} \frac{w_i}{t_i \ln 2} \sum_a \tr{\rho_{Q_AQ_BQ_E}\bigg(M_{a|x^*} \otimes \id_{Q_B}  \otimes (Z_{a,i} + Z_{a,i}^* + (1-t_i) Z_{a,i}^* Z_{a,i}) + t_i (\id_{Q_AQ_B} \otimes Z_{a,i} Z_{a,i}^*)\bigg)} \\
		&\mathrm{s.t.} \quad \sum_{abxy} c_{abxyj} \tr{\rho_{Q_AQ_BQ_E}(M_{a|x} \otimes N_{b|y} \otimes \id)} \geq v_j \hspace{3cm} \text{for all } 1\leq j \leq r  \\
		&\qquad \quad \sum_a M_{a|x} = \sum_b N_{b|y} = \id \hspace{6.35cm} \text{for all } x,y\\ 
		&\qquad \quad M_{a|x} \geq 0 \hspace{8.65cm}\text{for all } a,x \\
		&\qquad \quad N_{b|y} \geq 0 \hspace{8.75cm}\text{for all } b,y \\
		&\qquad \quad \|Z_{a,i}\| \leq \alpha_i \hspace{8.3cm}\text{for all }a, i=1,\dots,m-1 \\
		&\qquad \quad M_{a|x} \in B(Q_A), \quad N_{b|y}\in B(Q_B),\quad Z_{a,i} \in B(Q_E) \hspace{2.55cm}\text{for all } a,b,x,y,i\\
		&\qquad \quad \rho_{Q_AQ_BQ_E} \in D(Q_AQ_BQ_E)
	\end{aligned}
\end{equation}
where the infimum is taken over the all collections $(Q_AQ_BQ_E ,\rho_{Q_AQ_BQ_E}, \{M_{a|x}\}, \{N_{b|y}\}, \{Z_{a,i}\})$ satisfying the constraints of the problem. Note that when rearranging the objective function we have used the fact that the inner summation in~\eqref{eq:H_lower_bound} commutes with the infimum over the $Z_a$ operators. We can then further pull the infimum outside the outer summation by reparametrizing the variables as $Z_{a,i}$ for each $i$ in the outer sum. In order to compute a lower bound on~\eqref{eq:full_di_optimization_tensor} we employ the NPA hierarchy to relax this problem to an SDP. To do this we first drop the tensor product structure and instead include commutation relations on the relevant variables of the problem. In doing so we end up with the following noncommutative polynomial optimization problem which gives a lower bound on \eqref{eq:full_di_optimization_tensor}
\begin{equation}\label{eq:full_di_optimization_commuting}
	\begin{aligned}
		c_m + &\inf \quad \sum_{i=1}^{m-1} \frac{w_i}{t_i \ln 2}  \sum_a  \bra{\psi}M_{a|x^*} (Z_{a,i} + Z_{a,i}^* + (1-t_i) Z_{a,i}^* Z_{a,i}) + t_i Z_{a,i} Z_{a,i}^* \ket{\psi} &\\
		&\mathrm{s.t.} \quad \sum_{abxy} c_{abxyj} \bra{\psi} M_{a|x} N_{b|y} \ket{\psi} \geq v_j \hspace{3.8cm} \text{for all } 1\leq j \leq r  \\
		&\qquad \quad \sum_a M_{a|x} = \sum_b N_{b|y} = \id \hspace{4.6cm} \text{for all } x,y\\ 
		&\qquad \quad M_{a|x} \geq 0 \hspace{6.9cm}\text{for all } a,x \\
		&\qquad \quad N_{b|y} \geq 0 \hspace{7.cm}\text{for all } b,y \\
		&\qquad \quad Z_{a,i}^* Z_{a,i} \leq \alpha_i^2 \hspace{6.3cm}\text{for all }a, i=1,\dots,m-1 \\
		&\qquad \quad Z_{a,i} Z_{a,i}^* \leq \alpha_i^2 \hspace{6.3cm}\text{for all }a, i=1,\dots,m-1 \\
		&\qquad \quad [M_{a|x}, N_{b|y}] = [M_{a|x}, Z_{b,i}^{(*)}] = [N_{b|y}, Z_{a,i}^{(*)}] = 0 \hspace{1.55cm} \text{for all } a,b,x,y,i \\
		&\qquad \quad M_{a|x}, N_{b|y}, Z_{a,i} \in B(H) \hspace{4.7cm} \text{for all } a,b,x,y,i
	\end{aligned}
\end{equation}
where we have recalled that it is sufficient to consider pure states. Note that in both of the above optimizations we can also include projective measurement constraints without loss of generality. Using the NPA hierarchy we can then relax this optimization to a sequence of SDPs that give us a converging sequence of lower bounds on the optimal value and in turn a lower bound on the rate of the protocol. 

\begin{remark}[Commuting operator versus tensor product strategies]
	It is immediate that \eqref{eq:full_di_optimization_commuting} is never larger than \eqref{eq:full_di_optimization_tensor} and thus our subsequent relaxations of \eqref{eq:full_di_optimization_commuting} will always give lower bounds on the rates of protocols as we defined them previously using the tensor product framework. However, due to recent work~\cite{JNVY20} it may not be the case that \eqref{eq:full_di_optimization_tensor} and \eqref{eq:full_di_optimization_commuting} are equal.
	
	Let $R_m$ be the optimal value in \eqref{eq:full_di_optimization_commuting} and $r_{m,k}$ be the optimal value of its $k^{\mathrm{th}}$ level NPA relaxation. As we have explicit bounds on the operator norms of the variables we know that our NPA relaxations of \eqref{eq:full_di_optimization_commuting} will converge to the optimal value of \eqref{eq:full_di_optimization_commuting}, i.e., $R_m = \lim_{k \to \infty} r_{m,k}$. We believe that $\sup_m R_m$ will correspond to infimum of the conditional von Neumann entropy for commuting operator strategies. However proving this would require one to formally define \emph{commuting operator strategies} (analogous to the tensor product strategies introduced earlier) and checking that \eqref{eq:full_di_optimization_commuting} can be derived in a similar fashion to how~\eqref{eq:full_di_optimization_tensor} was derived. We leave this to future work. 
\end{remark}

\begin{remark}[Faster lower bounds]\label{rem:faster_computations}
	There are several ways to speed up the SDP relaxations of~\eqref{eq:full_di_optimization_commuting}. We will also note in the caption of each figure which speedups were used.
	\begin{enumerate}
		\item Often including the operator inequalities $Z_{a,i}^* Z_{a,i} \leq \alpha_i^2$ and $Z_{a,i}Z_{a,i}^* \leq \alpha_i^2$ does not improve the lower bound. Hence we always removed them when performing the computations. 
		\item The choice of monomial indexing set for the moment matrices can greatly affect the accuracy and speed of the computations. We found that the local level $1$ set, i.e., monomials of the form $ABZ$ where $A \in \{\id\}\cup\{M_{a|x}\}_{a,x}$, $B \in \{\id\}\cup\{N_{b|y}\}_{b,y}$ and $Z \in \{\id\} \cup \{Z_{c,i}, Z_{c,i}^*\}_{c,i}$ performed particularly well.
		\item It is also possible to commute the outer summation and the infimum, i.e., compute
		\begin{equation}
		\sum_{i=1}^{m-1} \frac{w_i}{t_i \ln 2}  \inf \sum_a  \bra{\psi}M_{a|x^*} (Z_{a,i} + Z_{a,i}^* + (1-t_i) Z_{a,i}^* Z_{a,i}) + t_i Z_{a,i} Z_{a,i}^* \ket{\psi}.
		\end{equation}
	 This can only decrease the lower bound on the entropy and hence it is sufficient for the purpose of lower bounding the rate of a protocol. The main advantage of doing so is that it drastically reduces the number of variables in the NPA hierarchy relaxations. Rather than running a single SDP with a $Z_{a,i}$ variable for each value of $a$ and $i$, we can instead run $m$ much smaller SDPs with only a $Z_{a,i}$ variable for each $a$. This significantly reduces the runtime of the SDPs and results in the runtime scaling linearly with the number of nodes in the Gauss-Radau quadrature. However we did notice in certain cases that the lower bounds computed in this manner were not converging to tight lower bounds. In such cases we did not include this speedup.
		\item The moment matrix of the NPA relaxation can without loss of generality be taken to be a real symmetric matrix. 
	\end{enumerate}
\end{remark}

\subsection{Numerical results}
We will now apply our method of computing rates to several DI-RE and DI-QKD scenarios and compare our technique with known analytical results~\cite{PABGMS, GMKB21b, WAP20} and other numerical techniques~\cite{TSGPL19, BFF21, BRC21}. We will only concern ourselves with the asymptotic rates in this work. However, as noted earlier, our technique can be combined with the entropy accumulation theorem in a relatively straightforward manner, similar to~\cite{BRC}, and thus one could also use it to compute finite round rates for protocols. We discuss this further in Appendix~\ref{app:eat}. The semidefinite relaxations were generated using the python package NCPOL2SDPA~\cite{ncpol2sdpa} and the resulting SDPs were solved using MOSEK~\cite{mosek}. As NCPOL2SDPA is no longer maintained, we used a maintained fork of the original package~\cite{ncpol2sdpa-fork}. We also provide example python scripts that implement some of the computations, these can be found at the github repository~\cite{github_dirates}.

\subsubsection{Randomness from the CHSH game}

To begin we consider the simplest possible setting of bounding $H(A|X=0,Q_E)$ when Alice and Bob's devices are constrained to achieve some minimal score in the CHSH game. In Figure~\ref{fig:chsh_comparison} we demonstrate how our bounds improve as we increase the number of nodes $m$ in the Gauss-Radau quadrature. We compare this with a known tight analytical bound in this setting~\cite{PABGMS}. In the figure we see that for $m=8$ our numerical technique effectively recovers the known tight analytical bound. As far as we are aware this is the first numerical technique to do so without resorting to algebraic simplifications afforded by Jordan's lemma. Importantly, this also demonstrates that in certain settings our technique can converge very quickly in the NPA hierarchy and in the size of the Gauss-Radau quadrature. Furthermore, as we are able to run our computations at low levels of the NPA hierarchy, the computations are also relatively fast for this setting with each SDP taking less than a second to run. We include additional plots demonstrating the convergence of our technique for other known tight analytical bounds~\cite{GMKB21b,WAP20} in Appendix~\ref{app:additional_plots}. These include multipartite scenarios useful for bounding the rates of DI conference key agreement protocols~\cite{MGKB20}.

\begin{figure}
	\centering
	\definecolor{mycolor2}{rgb}{0.00000,0.44700,0.74100}%
	\definecolor{mycolor4}{rgb}{0.85000,0.32500,0.09800}%
	\definecolor{mycolor3}{rgb}{0.92900,0.69400,0.12500}%
	\definecolor{mycolor1}{rgb}{0.49400,0.18400,0.55600}%
	\begin{tikzpicture}
		
		\begin{axis}[%
			width=4in,
			height=2.8in,
			scale only axis,
			xmin=0.75,
			xmax=0.854,
			ymin=0,
			ymax=1,
			grid=major,
			xlabel={CHSH score},
			ylabel={Bits},
			xtick={0.75, 0.77, 0.79, 0.81, 0.83, 0.85},
			axis background/.style={fill=white},
			legend style={at={(0.5,0.95)},legend cell align=left, align=left, draw=white!15!black}
			]
			\addplot[color=mycolor4, line width=0.9pt, mark=*, mark repeat=2, only marks] table[col sep=comma] {chsh_analytic_tikz.dat};
			\addlegendentry{$H(A|X=0,Q_E)$ analytic}	
			\addplot[color=mycolor1, line width=0.9pt] table[col sep=comma] {kosaki_chsh_local_2M_tikz.dat};
			\addlegendentry{Our technique $m=2$}
			\addplot[smooth, color=mycolor3, line width=0.9pt] table[col sep=comma] {kosaki_chsh_local_4M_tikz.dat};
			\addlegendentry{Our technique $m = 4$}
			\addplot[color=mycolor2, line width=0.9pt] table[col sep=comma] {kosaki_chsh_local_8M_tikz.dat};
			\addlegendentry{Our technique $m = 8$}	
		\end{axis}
	\end{tikzpicture}%
	\caption{\textbf{Recovering the local CHSH bound.} Comparison of lower bounds on $H(A|X=0,Q_E)$ for quantum devices that constrained to achieve some minimal CHSH score. Numerical bounds were computed using speedups (1), (3) and (4) from Remark~\ref{rem:faster_computations} at a relaxation level $2 + ABZ + AZZ$. A single SDP took less than one second to run.}
	
	\label{fig:chsh_comparison}
\end{figure}
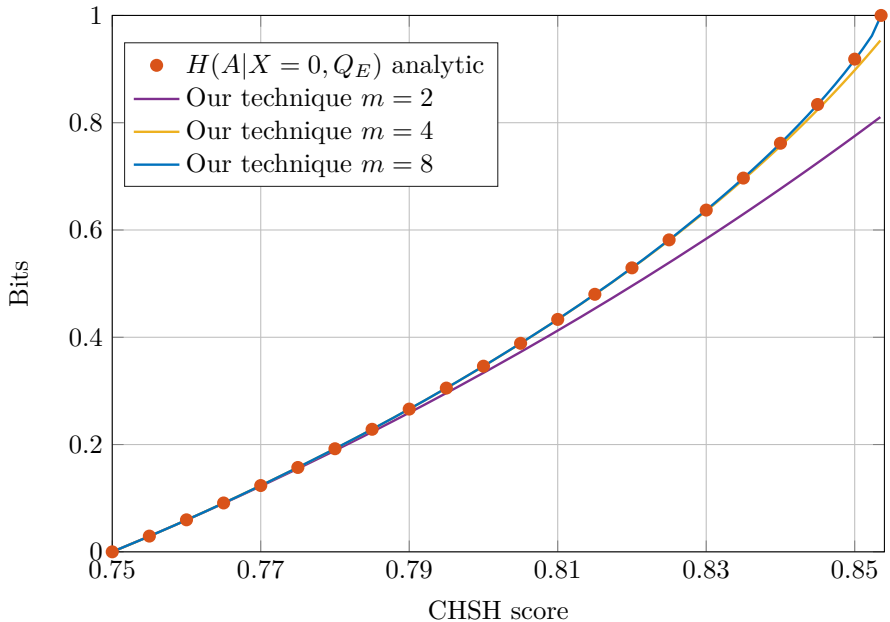  

For randomness expansion protocols it is beneficial to consider the randomness generated from both devices. Therefore, in Figure~\ref{fig:chsh_global_comparison} we plot lower bounds on $H(AB|X=0,Y=0,Q_E)$. We compare this again with the analytical bound on $H(A|X=0,Q_E)$ and we also compare it with the numerical technique of~\cite{BFF21} and a known upper bound from~\cite{BRC21}. As is clear from the figure, the global entropy can be significantly larger than the local entropy leading to much better rates for randomness expansion protocols based on the CHSH game. In comparison with the numerical technique of~\cite{BFF21}, we see that our new method vastly outperforms it. In the plot we also compare with some numerical upper bounds from~\cite{BRC21} which used Jordan's lemma to reduce the problem to qubits and then optimized over explicit qubit strategies. As we see from the plot our rate curve almost coincides with this upper bound and we expect that by increasing $m$ further the gap would be further reduced. Note that we could also have used speedup (3) from Remark~\ref{rem:faster_computations} in this case which would change the runtime from hours to minutes. However, we found that when using speedup (3) we were unable to recover the bound from~\cite{BRC21} and so we elected to not use it.

\begin{figure}
	\centering
	\definecolor{mycolor2}{rgb}{0.00000,0.44700,0.74100}%
	\definecolor{mycolor4}{rgb}{0.85000,0.32500,0.09800}%
	\definecolor{mycolor3}{rgb}{0.92900,0.69400,0.12500}%
	\definecolor{mycolor1}{rgb}{0.49400,0.18400,0.55600}%
	\begin{tikzpicture}
		
		\begin{axis}[%
			width=4in,
			height=2.8in,
			scale only axis,
			xmin=0.75,
			xmax=0.854,
			ymin=0,
			ymax=1.61,
			grid=major,
			xlabel={CHSH score},
			ylabel={Bits},
			xtick={0.75, 0.77, 0.79, 0.81, 0.83, 0.85},
			axis background/.style={fill=white},
			legend style={at={(0.57,0.95)},legend cell align=left, align=left, draw=white!15!black}
			]
			\addplot[color=mycolor2, line width=1.2pt] table[col sep=comma] {chsh_global_8M_full_tikz.dat};
			\addlegendentry{Our technique $m = 8$}
			\addplot[color=mycolor4, line width=1.2pt] table[col sep=comma] {chsh_analytic_tikz.dat};
			\addlegendentry{$H(A|X=0,Q_E)$ analytic}	
			\addplot[color=mycolor3, line width=1.2pt] table[col sep=comma] {h43_global_chsh_tikz.dat};
			\addlegendentry{$H^{\uparrow}_{(4/3)}(AB|X=0,Y=0,Q_E)$}
			\addplot[line width=1pt, dashed] table[col sep=comma] {roger_global_chsh_tikz.dat};
			\addlegendentry{Upper bound}	
		\end{axis}
	\end{tikzpicture}%
	\caption{\textbf{Global randomness from the CHSH game.} Comparison of lower bounds on $H(AB|X=0,Y=0,Q_E)$ for quantum devices that constrained to achieve some minimal CHSH score. Our technique was computed using speedups (1) and (2) from Remark~\ref{rem:faster_computations}.  For $m=8$ a single data point can take hours to run. We also compare with the iterated mean entropy $\imHup{4/3}(AB|X=0,Y=0,Q_E)$ from~\cite{BFF21} and an upper bound from~\cite{BRC21}.}
	
	\label{fig:chsh_global_comparison}
\end{figure}
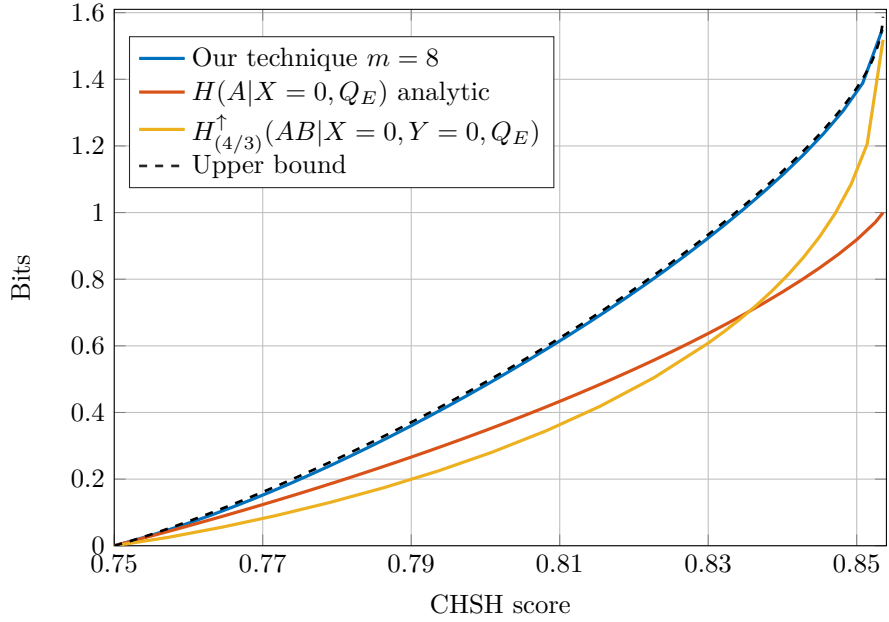

\subsubsection{Randomness from the full distribution}

By increasing our knowledge about the conditional distribution with which Alice and Bob's devices operate we further constrain the possible devices that could produce the statistics we observe in the protocol. In turn we can hope that this leads to larger bounds on the randomness produced by the devices. At the extreme end of this, we would have knowledge of the entire distribution characterizing the devices. In the following we assume that we have access to the complete distribution and we impose these constraints in the SDP. 

To make the plots more experimentally relevant we also assume the devices are affected by inefficient detection events. That is, with some probability $\eta \in [0,1]$ the devices will operate normally and with probability $1-\eta$ the devices will fail and deterministically output $0$. The conditional distribution for such devices takes the form 
\begin{equation}
	p(a,b|x,y) = \begin{cases}
		\eta^2 q(a,b|x,y) + \eta(1-\eta)(q(a|x)+q(b|y)) + (1-\eta)^2 \qquad&\text{if }a=0=b \\
		\eta^2 q(a,b|x,y) + \eta(1-\eta)q(a|x) &\text{if }a\neq0=b \\
		\eta^2 q(a,b|x,y) + \eta(1-\eta)q(b|y) &\text{if }a=0\neq b \\
		\eta^2 q(a,b|x,y) &\text{otherwise}
	\end{cases}
\end{equation}
where $q(a,b|x,y)$ is the conditional distribution of the devices when $\eta = 1$.

\begin{figure}
	\centering
	\definecolor{mycolor2}{rgb}{0.00000,0.44700,0.74100}%
	\definecolor{mycolor4}{rgb}{0.85000,0.32500,0.09800}%
	\definecolor{mycolor3}{rgb}{0.92900,0.69400,0.12500}%
	\definecolor{mycolor1}{rgb}{0.49400,0.18400,0.55600}%
	\begin{tikzpicture}
		
		\begin{axis}[%
			width=4in,
			height=2.8in,
			scale only axis,
			xmin=0.7,
			xmax=1.0,
			ymin=0,
			ymax=2.0,
			grid=major,
			xlabel={Detection efficiency ($\eta$)},
			ylabel={Bits},
			xtick={0.7, 0.75, 0.8, 0.85, 0.9, 0.95,1.0},
			ytick={0.0,0.2,0.4,0.6,0.8,1,1.2,1.4,1.6,1.8,2.0},
			axis background/.style={fill=white},
			legend style={at={(0.55,0.95)},legend cell align=left, align=left, draw=white!15!black}
			]
			\addplot[color=mycolor2, line width=1.2pt] table[col sep=comma] {kosaki_de_global_8M_tikz.dat};
			\addlegendentry{Our technique $m = 8$}
			\addplot[color=mycolor1, line width=1.2pt] table[col sep=comma] {ernest_global_de_tikz.dat};
			\addlegendentry{TSGPL bound}
			\addplot[color=mycolor3, line width=1.2pt] table[col sep=comma] {h2_global_DE_l2_tikz.dat};
			\addlegendentry{$H^{\uparrow}_{(2)}(AB|X=0,Y=0,Q_E)$}	
		\end{axis}
	\end{tikzpicture}%
	\caption{\textbf{Global randomness in the 2222-scenario.} Comparison of lower bounds on $H(AB|X=0,Y=0,Q_E)$ for quantum devices that are constrained by a full distribution. The numerical bounds for our technique were computed using speedups (1) and (3) from Remark~\ref{rem:faster_computations} at relaxation level $2 + ABZ$ including all monomials present in the objective function. A single SDP takes less than a minute to run at this level. We also compare with the iterated mean entropy $\imHup{2}(AB|X=0,Y=0,Q_E)$ from~\cite{BFF21} and the numerical technique from~\cite{TSGPL19} which we refer to as the TSGPL bound.}
	
	\label{fig:de_global_2222}
\end{figure}
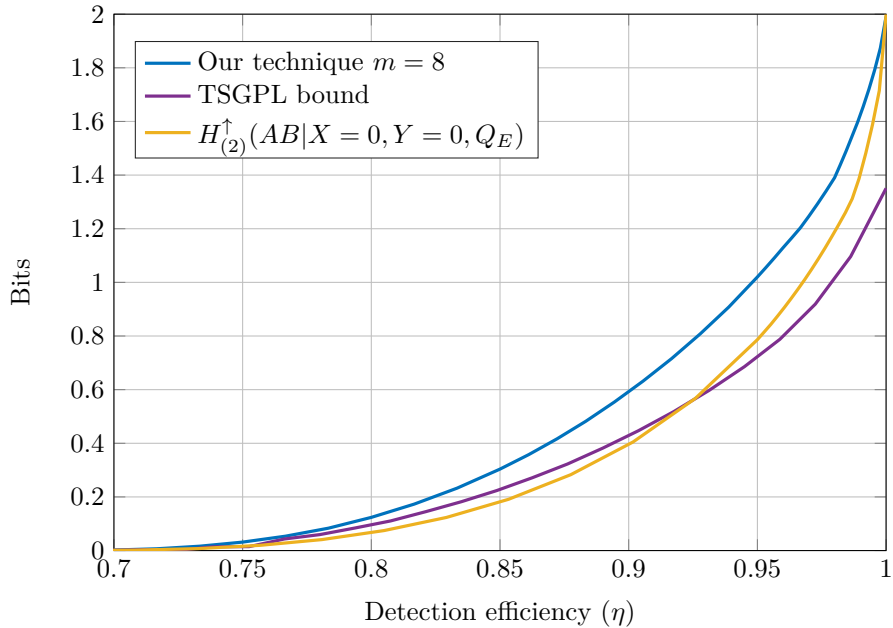

In Figure~\ref{fig:de_global_2222} we compare lower bounds on the entropy $H(AB|X=0,Y=0,Q_E)$  when Alice and Bob's devices have inefficient detectors. For the curve representing our technique we used a Gauss-Radau decomposition with $m=8$ nodes and at each data point we selected a two-qubit distribution in order to maximize the entropy produced by the devices. We compare our technique with the iterated mean entropy $\imHup{2}(AB|X=0,Y=0,Q_E)$ from~\cite{BFF21} and the TSGPL method from~\cite{TSGPL19} both of which also constrained the devices by their full distribution. Compared with the other methods, our technique is everywhere larger and on the whole the difference is substantial. We also note that our technique is again significantly faster than the other numerical techniques presented here.

\subsubsection{Better bounds on DI-QKD key rates}

Thus far we have only concerned ourselves with bounds on local and global entropies. However, we can use these bounds to compute the asymptotic rates of some DI-QKD protocols. The asymptotic rate of a DI-QKD protocol with one-way error correction is given by the Devetak-Winter bound~\cite{DW05}
\begin{equation}
H(A|X=x^*,Q_E) - H(A|B,X=x^*,Y=y^*),
\end{equation}
where again we are assuming a spot checking protocol and we consider the rate as the spot checking probability tends to 0 and the number of rounds tends to infinity. The second term in the rate can be directly estimated from the statistics in the protocol and the first term can be lower bounded using our technique.

We consider DI-QKD protocols in the 2322-scenario and we take $(x^*,y^*) = (0,2)$, i.e., Bob's third input acts as his key generating input. In the same setup as the previous figure we consider devices with inefficient detectors and we constrain them by their full distribution.\footnote{Actually, we only constrain the devices by their distribution for inputs $(X,Y) \in \{0,1\}^2$.} However, for QKD we allow Bob to record a device failure with an additional symbol $\perp$ when he receives his key generating input $y^*=2$. By doing this, he collects more detailed information about his device's behaviour and in turn this allows him to reduce the size of $H(A|B,X=x^*,Y=y^*)$ slightly. We refer the reader to~\cite{ML12} for a more detailed discussion of this post-processing of no-click events. 

To further boost the rates we also include preprocessing of the raw key~\cite{HSTRBS20}. Loosely, we allow for Alice to add additional noise to the outputs of her devices. In certain circumstances, this can increase the value of $H(A|X=x^*, Q_E)$ more than it increases the value of $H(A|B,X=x^*,Y=y^*)$ and therefore increasing the rate overall. More specifically, after Alice and Bob have collected their raw key (the outputs of their devices when $(X,Y) = (0,2)$) Alice will independently flip each of her key-bits with some fixed probability $q \in [0,1/2]$. For example, the post-measurement state for a single key-generating round is transformed after this preprocessing to 
\begin{equation}
\begin{aligned}
	\rho_{ABQ_E} &= \sum_{ab} ((1-q)\outer{a} + q \outer{a \oplus 1}) \otimes \outer{b} \otimes \rho_{Q_E}(a,b,x^*,y^*) \\
	&= \sum_{ab} \outer{a} \otimes \outer{b} \otimes ((1-q)\rho_{Q_E}(a,b,x^*,y^*) + q\,\rho_{Q_E}(a\oplus 1,b,x^*,y^*)) \\
	&= \sum_{ab} \outer{a} \otimes \outer{b} \otimes \ptr{Q_AQ_B}{\rho_{Q_AQ_BQ_E} (((1-q) M_{a|x^*} + q M_{a\oplus 1 |x^*}) \otimes N_{b|y^*} \otimes \id)}.
\end{aligned}
\end{equation}
Thus this preprocessing can be seen as equivalent to Alice transforming her measurement from $\{M_0, M_1\}$ to $\{(1-q) M_0 + q M_1, (1-q) M_1 + q M_0\}$ on key generating rounds. It follows that we can then model this preprocessing in our numerical computations by modifying Alice's measurement operators in the objective function of~\eqref{eq:full_di_optimization_tensor} appropriately. 
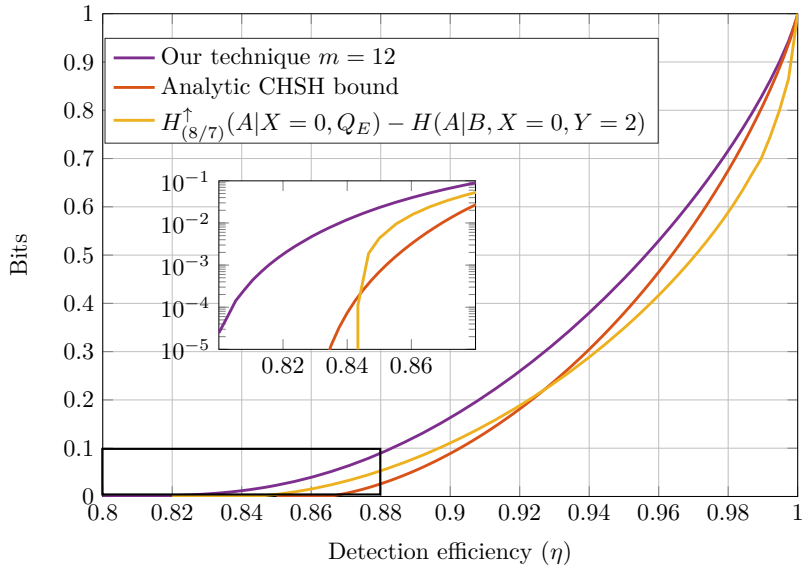
\begin{figure}
	\centering
	\definecolor{mycolor2}{rgb}{0.00000,0.44700,0.74100}%
	\definecolor{mycolor4}{rgb}{0.85000,0.32500,0.09800}%
	\definecolor{mycolor3}{rgb}{0.92900,0.69400,0.12500}%
	\definecolor{mycolor1}{rgb}{0.49400,0.18400,0.55600}%
	\begin{tikzpicture}[scale=0.9]
		
		\begin{axis}[%
			width=4in,
			height=2.8in,
			scale only axis,
			xmin=0.8,
			xmax=1.0,
			ymin=0,
			ymax=1.0,
			grid=major,
			xlabel={Detection efficiency ($\eta$)},
			ylabel={Bits},
			xtick={0.8, 0.82, 0.84, 0.86, 0.88, 0.9, 0.92, 0.94, 0.96, 0.98,1.0},
			ytick={0.0,0.1,0.2,0.3,0.4,0.5,0.6,0.7,0.8,0.9,1.0},
			axis background/.style={fill=white},
			legend style={at={(0.8,0.95)},legend cell align=left, align=left, draw=white!15!black}
			]
			\addplot[color=mycolor1, line width=1.2pt] table[col sep=comma] {qkd_2322_12M_noisy_tikz.dat};
			\addlegendentry{Our technique $m = 12$}
			\addplot[color=mycolor4, line width=1.2pt] table[col sep=comma] {qkd_chsh_preprocessing_analytic_tikz.dat};
			\addlegendentry{Analytic CHSH bound}	
							\addplot[color=mycolor3, line width=1.2pt] table[col sep=comma] {h87_de_optbin_tikz.dat};
							\addlegendentry{$H^{\uparrow}_{(8/7)}(A|X=0,Q_E) - H(A|B,X=0,Y=2)$}
			\coordinate (insetPosition) at (axis cs:0.81,0.2);
			\coordinate (pos1big) at (axis cs:0.797,-0.012);
			\coordinate (pos2big) at (axis cs:0.869,0.073);
			\coordinate (topleftbig) at (axis cs:0.8,0.1);
			\coordinate (bottomrightbig) at(axis cs:0.88,0.0);	
			,	\end{axis}
		\begin{semilogyaxis}[at={(insetPosition)},anchor={outer south west},
			width=2.1in,
			height=1.6in,
			xmin=0.8,
			xmax=0.88,
			ymin=1e-5,
			ymax=1e-1,
			ytick = {1e-5, 1e-4, 1e-3, 1e-2, 1e-1},
			xtick = {0.82, 0.84, 0.86},
			xticklabels={0.82,0.84,0.86},
			yticklabels={$10^{-5}$,$10^{-4}$,$10^{-3}$,$10^{-2}$,$10^{-1}$},
			axis background/.style={fill=white}
			]
			\addplot[color=mycolor1, line width=1.2pt] table[col sep=comma] {qkd_2322_12M_noisy_tikz.dat};
			\addplot[color=mycolor4, line width=1.2pt] table[col sep=comma] {qkd_chsh_preprocessing_analytic_tikz.dat};
			\addplot[color=mycolor3, line width=1.2pt] table[col sep=comma] {h87_de_optbin_tikz.dat};
			\coordinate (topleftsmall) at (axis cs:0.79,0.04);
			\coordinate (bottomrightsmall) at (axis cs:0.87,-0.01);
		\end{semilogyaxis}
		\draw[thick] (pos1big) rectangle (pos2big); 
	\end{tikzpicture}%
	\caption{\textbf{QKD 2322 rates.} Comparison of lower bounds on the key rate $H(A|X=0,Q_E) - H(A|B,X=0,Y=2)$. The curves related to our technique were computed using speedups (1) and (3) from Remark~\ref{rem:faster_computations} at a relaxation level $2 + ABZ$ and a single SDP at this relaxation level takes a few seconds to run. We also compare with the iterated mean entropy $\imHup{8/7}$ from~\cite{BFF21} and with the analytic bound for a CHSH based DI-QKD protocol with a preprocessing step from~\cite{HSTRBS20,WAP20}.}
	
	\label{fig:qkd_2322}
\end{figure}

In Figure~\ref{fig:qkd_2322} we plot the key rates achievable for devices with inefficient detectors. We compute our technique using a Gauss-Radau quadrature of $m=12$ nodes. We compare our technique to bounds on the rate given by the iterated mean entropy $\imHup{8/7}$ from~\cite{BFF21} (we do not include preprocessing for this curve as it did not give improvements) and an analytical bound from~\cite{HSTRBS20,WAP20} which is for a protocol based on the CHSH game and includes the preprocessing step. For both of the curves that incorporate the preprocessing, at each data point we optimized the probability $q$ with which Alice performs the her bitflip in order to maximize the obtained rates. We see that for the entirety of the plot the curve computed using our technique outperforms the other rate curves. In particular this shows the advantage that one can gain by changing the knowledge collected about the devices from a Bell-inequality violation to the full distribution. In the inset plot we zoom in on the region $[0.8, 0.88] \times [10^{-5},0.1]$. Looking at the inset plot we can inspect where the various rate curves vanish, in other words the minimal detection efficiency required to perform DI-QKD according to that curve. For the curve computed using the iterated mean entropy $\imHup{8/7}$ the minimal detection efficiency is around $0.845$. The red curve based on the analytic bound actually vanishes around $0.826$~\cite{WAP20}. However, we see that by using our method (purple curve) we are able to substantially reduce the threshold detection efficiency down to under $0.8$. This threshold is now within a regime that is experimentally achievable. This warrants a more thorough analysis, using a more realistic model and computing finite round rates, in order to ascertain the performance capabilities of a current DIQKD experiment. We leave such an analysis to future work.

\section{Methods}\label{sec:methods}

The objective of this section is to derive the variational upper bounds on the relatively entropy that were used in the previous section (see Theorem~\ref{thm:D_upper_bounds} below). However, first we need to define the quantum relative entropy and more generally quasi-relative entropies in the framework of von Neumann algebras and establish some of their properties (Section~\ref{sec:quasi-rel}). We will then also describe rational approximations of the logarithm function (Section~\ref{sec:approx_log}), which are a crucial ingredient in the derivation.

\subsection{Quasi-relative entropies}
\label{sec:quasi-rel}

Quasi-relative entropies can be defined in general for two positive semidefinite linear functionals on a von Neumann algebra $\cA$. A linear functional $\rho : \cA \to \CC$ is said to be \emph{positive semidefinite} if $\rho(a^* a) \geq 0$ for all $a \in \cA$. We refer to a positive semidefinite linear functional satisfying $\rho(\id) = 1$ as a \emph{state}. We will mostly focus on the setting where $\cA = B(H)$ for some separable Hilbert space $H$ and positive semidefinite linear functionals $\rho$ defined by some trace-class operator $\widetilde{\rho}$ by $\rho(a) = \tr{\widetilde{\rho} a}$ for all $a \in B(H)$. To simplify notation, we will often use the same symbol $\rho$ for both the positive semidefinite linear functional $\rho$ and the trace-class positive semidefinite operator $\widetilde{\rho}$.

\paragraph{Functional calculus of quadratic forms} We review briefly the functional calculus of sesquilinear (or quadratic) forms introduced by Pusz and Woronowicz in \cite{PW75} and further applied in \cite{PW78}. Let $\alpha$ and $\beta$ be two positive semidefinite \emph{quadratic forms} on a complex vector space $U$. A function $\alpha : U \to \RR$ is called a quadratic form if it satisfies $\alpha(\lambda u) = |\lambda|^2 \alpha(u)$ for any $\lambda \in \CC$ and $u \in U$ as well as the parallelogram identity $\alpha(u+v) + \alpha(u-v) = 2\alpha(u) + 2\alpha(v)$.  We say that it is positive semidefinite if $\alpha(u) \geq 0$ for all $u \in U$. The theory of Pusz and Woronowicz allows one to define a new quadratic form $F(\alpha,\beta)$ for any function $F:\RR_+^2 \to \RR$ that is measurable (with respect to the $\sigma$-algebra of Borel subsets of $\RR_+^2$), \emph{positive homogeneous} (i.e., $F(\lambda x , \lambda y) = \lambda F(x,y)$ for all $\lambda \geq 0$) and \emph{locally bounded from below} (i.e., bounded from below on all compact sets). For example if $F(x,y) = \sqrt{xy}$, this theory allows us to define the geometric mean of two positive semidefinite quadratic forms.

The main idea of the theory is that one can always represent a pair of positive semidefinite quadratic forms by two positive semidefinite, \emph{commuting} operators. We call a \emph{representation} of $(\alpha,\beta)$ a tuple $(H,A,B,h)$ where $H$ is a Hilbert space, $A,B$ are positive semidefinite and commuting bounded operators on $H$, and $h:U\to H$ is a linear map onto a dense subset of $H$ such that $\alpha(u) = \<h(u), A h(u) \>_H$, $\beta(u) = \<h(u), B h(u) \>_H$ for all $u \in U$. As shown in \cite[Theorem 1.1]{PW75}, such a representation always exists. Then one defines the quadratic form $F(\alpha,\beta)$ on $U$ by
\begin{equation}
\label{eq:def_joint_spectral_measure}
F(\alpha,\beta)(u) = \< h(u), F(A,B) h(u) \>_H = \int_{\RR_+^2} F(x,y) \< h(u), \mathrm{d}E(x,y) h(u) \> \ ,
\end{equation}
where $E$ is the joint spectral measure of the commuting operators $A$ and $B$; see~\cite[Chapter 5]{schmudgen2012unbounded} for the definition of the joint spectral measure and \cite{PW78} for the justification of the existence of this integral. The main result of the theory \cite[Theorem 1.2]{PW75} is that the quadratic form $F(\alpha, \beta)$ defined above is independent of the choice of representation $(H,A,B,h)$ of the pair $(\alpha,\beta)$.

\paragraph{Defining quasi-relative entropies using the functional calculus of quadratic forms} Following the work of \cite{PW78}, we define the quasi-relative entropy of two positive semidefinite linear functionals $\rho,\sigma$ on a von Neumann algebra $\cA$.
Given $\rho$ and $\sigma$, we define two positive semidefinite quadratic forms on $\cA$ (note that $\cA$ is also a complex vector space) by 
\begin{equation}
	\label{eq:LrhoRrho}
	L_{\sigma}(a) = \sigma(aa^*) \quad \text{and} \quad R_{\rho}(a) = \rho(a^* a) \quad \text{for } a \in \cA \ .
\end{equation}
Then using the functional calculus of Pusz-Woronowicz, one can define $F(L_{\sigma},R_{\rho})$, as a quadratic form on $\cA$. This leads to the following definition for the $F$-quasi-relative entropy.
\begin{definition}[$F$-quasi-relative entropy]
	\label{def:quasi-rel-entr}
	Let $\rho$ and $\sigma$ be two positive semidefinite linear functionals on the von Neumann algebra $\cA$. Then for any measurable, positive homogeneous and locally bounded from below function $F: \RR^2_+ \to \RR \cup \{+\infty\}$ the $F$-quasi-relative entropy between $\rho$ and $\sigma$ is
	\begin{equation}
		D_F(\rho \| \sigma) := F(L_{\sigma},R_{\rho})(\id)
	\end{equation}
	where 
	\begin{equation}
		\label{eq:pwdiv}
		F(L_{\sigma}, R_{\rho})(I) := \< h(I) , F(A,B) h(I) \>_K \,
	\end{equation}
	with $(K,A,B,h)$ being a representation of the two positive semidefinite quadratic forms $(L_{\sigma},R_{\rho})$ and this expression is defined by the integral in~\eqref{eq:def_joint_spectral_measure}. We denote by $\nu_{\rho, \sigma}$ the positive measure $\nu_{\rho, \sigma} = \<h(I), E h(I) \>$ on $\RR_+^2$, where $E$ is the joint spectral measure of $A$ and $B$. Note that $\nu_{\rho, \sigma}$ only depends on $\rho$ and $\sigma$ and not on the function $F$. With this notation we have
	\begin{equation}
	\label{eq:def_div_int}
	 D_F(\rho \| \sigma) = \int_{\RR_+^2} F(x,y) \mathrm{d} \nu_{\rho, \sigma}(x,y) \ ,
	 \end{equation}
with $\int_{\RR_+^2} x \mathrm{d} \nu_{\rho, \sigma}(x,y) = \sigma(I)$ and $\int_{\RR_+^2} y \mathrm{d} \nu_{\rho, \sigma}(x,y) = \rho(I)$.
\end{definition}
We remark that for the applications we consider, it is important to allow the function $F$ to be infinite on some points. As described in~\cite{PW78}, the functional calculus can readily be extended to this setting. In this setup, the quadratic form $F(L_{\sigma},R_{\rho})$ can take the value $+\infty$ on some points and thus, as expected, the corresponding $F$-quasi-relative entropy can be infinite.

\begin{definition}[Quantum relative entropy, also called Umegaki divergence~\cite{umegaki}]
	Let $\rho$ and $\sigma$ be two positive semidefinite linear functionals on the von Neumann algebra $\cA$. Then the \emph{relative entropy} between $\rho$ and $\sigma$, written $D(\rho\|\sigma)$ is defined as the $F$-quasi-relative entropy with $F(x,y) = y \log_2(y/x)$:
	\begin{equation}
		D(\rho\| \sigma) :=	D_{F}(\rho\|\sigma) \,.
	\end{equation}
\end{definition}

Note that when $\cA = B(H)$ and $\sigma$ is the trace, we obtain the von Neumann entropy $H(\rho) = - D(\rho \| \sigma)$. The conditional von Neumann entropy can also be obtained in a similar way from the quantum relative entropy, this is described in Section~\ref{sec:DI}.

\begin{definition}[$\alpha$-quasi-relative entropy]
	\label{def:alpha-quasi}
	Let $\rho$ and $\sigma$ be two positive semidefinite linear functionals on $\cA$ and let $\alpha \in (0,1) \cup (1,\infty)$. Then the $\alpha$-quasi-relative entropy written $Q_{\alpha}(\rho \|\sigma)$ is defined as $F$-quasi-relative entropy with $F(x,y) = y^\alpha x^{1-\alpha}$.
\end{definition}

\begin{remark}[Finite-dimensional case and the relative modular operator]
	\label{rem:finitedimensions}
	It is instructive to consider the setting where $H$ is a $d$-dimension Hilbert space. In this case $B(H)$ is the algebra of $d\times d$ matrices with elements in $\CC$. Then the two positive semidefinite linear functionals $\rho$ and $\sigma$ can be represented by positive semidefinite operators $\widetilde{\rho}$ and $\widetilde{\sigma}$, i.e.,  $\rho(a) = \tr{\widetilde{\rho} a}$ and $\sigma(a) = \tr{\widetilde{\sigma} a}$ for all $a \in B(H)$.

	Let $\cL_{\sigma}$ and $\cR_{\rho}$ be left and right-multiplication operators by $\widetilde\sigma$ and $\widetilde{\rho}$, that is $\cL_{\sigma},\cR_{\rho}:B(H) \to B(H)$ where $\cL_{\sigma}(a) = \widetilde{\sigma} a$ and $\cR_{\rho}(a) = a \widetilde{\rho}$. Note that $\cL_{\sigma}$ and $\cR_{\rho}$ are two commuting operators on $B(H)$ that are self-adjoint with respect to the Hilbert-Schmidt inner product on $B(H)$, $\langle a , b \rangle_{HS} = \tr{a^* b}$. These operators realize the quadratic forms $L_{\sigma}$ and $R_{\rho}$ defined in \eqref{eq:LrhoRrho}, in the sense that $L_{\sigma}(a) = \<a , \cL_{\sigma} a\>_{HS}$ and $ R_{\rho}(a) = \<a, \cR_{\rho} a\>_{HS}$. Moreover, it is clear from their actions that $[\cL_{\sigma},\cR_{\rho}]=0$. In particular, our representation is $(B(H), \cL_\sigma, \cR_\rho, \mathrm{id})$ where $\mathrm{id}:B(H) \rightarrow B(H)$ is the identity map. Thus, for an $F$ satisfying the conditions of Definition~\ref{def:quasi-rel-entr}, the $F$-quasi-relative entropy between $\rho$ and $\sigma$ is given by
	\begin{equation}
	F(L_{\sigma},R_{\rho})(\id) = \<\id , F(\cL_{\sigma},\cR_{\rho}) \id\>_{HS}\, .
	\end{equation}
	This is precisely the definition used in \cite{effrospnas}. Indeed, if we introduce the relative modular operator $\Delta = \cL_{\sigma} \cR_{\rho}^{-1}$ and we define $f(x) := F(x,1)$ then for positive homogeneous $F$ we have $F(x,y) = y f(x/y)$ and the above is equal to
	\begin{equation}
	\<\id , \cR_{\rho} f(\cL_{\sigma} \cR_{\rho}^{-1})  \id\>_{HS} = \< \cR_{\rho}^{1/2} \id, f(\Delta) \cR_{\rho}^{1/2} \id \>_{HS} = \<\widetilde\rho^{1/2} , f(\Delta) \widetilde\rho^{1/2} \>_{HS}
	\end{equation}
	where we used the fact that $\cR_{\rho}^{1/2}$ is self-adjoint with respect to the Hilbert-Schmidt inner product, and that $\cR_{\rho}^{1/2} a = a \widetilde{\rho}^{1/2}$. Note that when $\rho$ is not invertible the operator $\Delta$ multiplies on the right by the generalized inverse of $\widetilde\rho$.
	It can be verified that in the finite-dimensional case, the quasi-relative entropy is given by (see e.g., \cite[Equation~15]{S14})
	\begin{equation}
	D_F(\rho \| \sigma) = \sum_{j,k} F(q_k,p_j) |\< \phi_k | \psi_j\>|^2
	\end{equation}
	where $\widetilde{\rho} = \sum_{j} p_j \proj{\psi_j}$ and $\widetilde{\sigma} = \sum_{k} q_k \proj{\phi_k}$ are spectral decompositions of the density matrices $\widetilde{\rho}$ and $\widetilde{\sigma}$ respectively. Finally, choosing $F(x,y) = y \log_2(y/x)$ we recover the usual expression for the quantum relative entropy
	\begin{equation}
	D(\rho\|\sigma) = \tr{\rho (\log_2 \rho - \log_2 \sigma)}
	\end{equation}
	and for the $\alpha$-quasi-relative entropies, $F(x,y) = x^{1-\alpha} y^\alpha$, we find $Q_\alpha(\rho\|\sigma) = \tr{\rho^{\alpha} \sigma^{1-\alpha}}$ which is the quantity within the logarithm of the Petz-R\'enyi divergences~\cite{Petz}.
\end{remark}

In order to obtain our variational expression for the quantum relative entropy, we will use a variational representation of $D_{F_t}$ when $F_t(x,y) = y\frac{x-y}{t(x-y)+y}$ for $t \in (0,1]$. The relevance of $F_t$ to the relative entropy is that if we let $F(x,y) = y\log_2(y/x)$, then $F(x,y) = -\frac{1}{\ln 2}\int_{0}^{1} F_t(x,y) dt$.

\begin{proposition}
	\label{prop:varparallelsum}
	Let $F_t(x,y) = y\frac{x-y}{t(x-y)+y}$ for $t \in (0,1]$, let $\rho$ and $\sigma$ be positive semidefinite linear functionals on a von Neumann algebra $\cA$. Then
	\begin{equation}
		\label{eq:var-DF-parallel}
		D_{F_t}(\rho\|\sigma) = \frac{1}{t} \inf_{a \in \cA} \left\{\rho(I) + \rho(a+a^*) + (1-t) \rho(a^* a) + t\sigma(aa^*)\right\} \ .
	\end{equation}

	Furthermore, if $t < 1$, and $\rho$ and $\sigma$ are two trace-class positive semidefinite operators on the separable Hilbert space $H$ satisfying $\rho \leq \lambda \sigma$ for some $\lambda \in \RR_+$, then the infimum in \eqref{eq:var-DF-parallel} is achieved at an element with norm bounded by $\alpha := \max( \frac{3}{2(1-t)}, \frac{3\lambda}{2t} )$, i.e., we can write
	\begin{align}
		\label{eq:var_expr_bound_norm}
		D_{F_t}(\rho\|\sigma) = \frac{1}{t} \inf_{\substack{Z \in B(H) \\ \| Z \| \leq \alpha}} \left\{\tr{\rho} + \tr{\rho (Z + Z^*)} + (1-t) \tr{\rho Z^* Z} + t \tr{\sigma ZZ^*} \right\} \ .
	\end{align}
\end{proposition}

\begin{remark}
Note that a constant bound on the norm of $Z$ in \eqref{eq:var_expr_bound_norm} can only be guaranteed when $t < 1$. Indeed, observe that for that $t=1$ we have $F_1(x,y) = y - y^2 / x$ so that $D_{F_1}(\rho \| \sigma) = \rho(I) - Q_2(\rho \| \sigma)$, and the variational expression becomes:
\begin{equation}
D_{F_1}(\rho \| \sigma) = \inf_{Z \in B(H)} \left\{ \tr{\rho} + \tr{\rho(Z+Z^*)} + \tr{\sigma ZZ^*} \right\}.
\end{equation}
Assuming $\rho$ and $\sigma$ are finite-dimensional operators the infimum can be shown to be attained at $Z = -\sigma^{-1} \rho$ with an optimal value equal to $\tr{\rho} - \tr{\rho^2 \sigma^{-1}}$. One can see that the operator norm of $\sigma^{-1} \rho$ can be made arbitrarily large even if we assume that $\rho \leq \lambda \sigma$ (e.g., it suffices to take $\rho$ to be a pure state).
\end{remark}

	\begin{proof}
		Proceeding in a similar way as~\cite[Lemma]{PW75}, if $A,B$ are two commuting Hermitian operators on a Hilbert space $H$, for any $y, z$ with $y + z = x$, a simple calculation gives
		\begin{equation}
		\< x , F_t(A,B) x \> = \frac{1}{t} \< x, \frac{t(A-B)B}{t(A-B) + B} x \> = \frac{1}{t} \left(\< y , t(A-B) y \> + \< z , B z \> - \< u, (t(A-B) + B) u \> \right) \ ,
		\end{equation}
		where $u = (t(A-B)+B)^{-1} B x - y$. As $t(A-B) + B$ is a positive semidefinite operator and we may choose $y$ so that $u = 0$, it follows that
		\begin{equation}
		\< x , F_t(A,B) x \> = \frac{1}{t} \min_{y+z = x} \< y , t(A-B) y \> + \< z , B z\> \ .
		\end{equation}
		Now we consider the positive semidefinite quadratic forms $L_\sigma$, $R_\rho$ for the positive semidefinite functionals $\sigma$ and $\rho$ as defined in~\eqref{eq:LrhoRrho} and we take a corresponding representation $(H,A,B,h)$. We get
		\begin{equation}
		F_t(L_{\sigma}, R_{\rho})(I) = \frac{1}{t} \min_{y+z = h(I)} \< y, t(A - B) y \> + \<z , B z \> \ .
		\end{equation}
		
		Since $h: \cA \to H$ is onto a dense subset of $H$, we can replace the minimum with an infimum to $y,z \in H$ such that $y = h(a)$ and $z=h(b)$. Then using the fact that $\< h(a), A h(a) \> = \sigma(aa^*)$ and $\< h(a), B h(a) \> = \rho(a^*a)$ we get
		\begin{equation}
		F_t(L_{\sigma}, R_{\rho})(I) = \frac{1}{t} \inf_{h(a)+h(b)=h(I)} \left\{ t\sigma(aa^*) - t\rho(a^* a) + \rho(b^* b)\right\}.
		\end{equation}
		Note that if we have $a_1$ and $a_2$ satisfying $h(a_1) = h(a_2)$, then $\sigma(a_1a_1^*) = \< h(a_1), A h(a_1)\> = \< h(a_2), A h(a_2)\> = \sigma(a_2 a_2^*)$ and similarly for $\rho$. As such we can replace the condition in the infimum by simply $a+b = I$, which gives
		\begin{equation}
		\begin{aligned}
			D_{F_t}(\rho\|\sigma) &= \frac{1}{t} \inf_{a+b=I} \left\{t\sigma(aa^*) - t\rho(a^* a) + \rho(b^* b)\right\}\\
			&= \frac{1}{t} \inf_{a \in \cA} \left\{\rho(I) - \rho(a+a^*) + (1-t) \rho(a^* a) + t\sigma(aa^*)\right\} \\
			&= \frac{1}{t} \inf_{a \in \cA} \left\{\rho(I) +  \rho(a+a^*) + (1-t) \rho(a^* a) + t\sigma(aa^*)\right\}.
		\end{aligned}
		\end{equation}
		where on the final line we made the substitution $a \mapsto -a$. This proves \eqref{eq:var-DF-parallel}.
		
		To prove the second part we first derive sufficient conditions for $a_0$ to achieve the infimum. Call $\phi(a)$ the expression to be infimized in~\eqref{eq:var-DF-parallel}. Note that $\phi$ is convex since $\rho$ and $\sigma$ are positive semidefinite (note for example that the restriction of $\phi$ to any line $a + s b$ is a convex quadratic in $s \in \RR$). If $a_0$ achieves the infimum of $\phi$ it must be that $\frac{d}{ds} \phi(a_0 + s b)|_{s=0} = 0$ for any $b \in \cA$. This equation gives
		\begin{equation}
		\label{eq:optim1}
		\rho(b+b^*) + (1-t) (\rho(a_0^* b + b^* a_0)) + t \sigma(a_0 b^* + b a_0^*) = 0 \;\; \forall b \in \cA.
		\end{equation}
		Since $\phi$ is convex this condition is also sufficient for optimality.
		
		Now assume that $\rho$ and $\sigma$ are trace-class positive semidefinite operators on the separable Hilbert space $H$ so that $\rho(a)=\tr{\rho a}$ and $\sigma(a) = \tr{\sigma a}$. Equation \eqref{eq:optim1} says that a necessary and sufficient condition for $A_0 \in B(H)$ to achieve the infimum is that
		\begin{equation}
		\tr{\rho(B+B^*)} + (1-t) \tr{\rho A_0^* B + \rho B^* A_0} + t \tr{\sigma A_0 B^* + \sigma B A_0^*} = 0 \;\; \forall B \in B(H).
		\end{equation}
Letting $M = \rho + (1-t) A_0 \rho + t \sigma A_0$, this is the same as $\tr{MB^*} + \tr{M^* B} = 0$ for all $B \in B(H)$ which is equivalent to $M = 0$. For convenience, we let $Z = -A_0$ and arrive at the following operator Sylvester equation
		\begin{equation}
		(1-t) Z \rho + t\sigma Z = \rho.
		\end{equation}
The existence of a bounded solution to this equation is guaranteed by the following lemma. 
\begin{lemma}
\label{lem:lyapeq}
Let $t \in (0,1)$ and let $\rho$ and $\sigma$ be two trace-class positive semidefinite operators on the separable Hilbert space $H$ satisfying $\rho \leq \lambda \sigma$ for some $\lambda \in \RR_+$. Then the operator equation
	\begin{align}
		\label{eq:equation_Z}
		(1-t) Z \rho + t \sigma Z = \rho
	\end{align}
	has a solution $Z \in B(H)$ which satisfies $\| Z \| \leq \alpha := \max\{ \frac{3}{2(1-t)}, \frac{3\lambda}{2t} \}$.
\end{lemma}
\begin{proof}
See Appendix \ref{sec:lyapeq}.
\end{proof}
This completes the proof of Proposition~\ref{prop:varparallelsum}.
\end{proof}

\subsection{Rational lower bounds on the logarithm}
\label{sec:approx_log}

The quantum relative entropy is defined as the $F$-quasi relative entropy for $F(x,y) = y \log_2 y - y \log_2 x$. In this section we define a sequence of rational upper bounds on $F$ that are expressed as a finite sum of the functions $F_t$ for $t \in (0,1]$.

The natural logarithm function (written $\ln$) has the following integral representation
\begin{equation}
	\label{eq:logint}
	\ln(x) = \int_{0}^{1} f(t,x) dt
\end{equation}
where
\begin{equation}
	\label{eq:def_ftx}
	f(t,x) = \frac{x-1}{t(x-1)+1}.
\end{equation}
To get a rational approximation of $\ln$, we can discretize the integral \eqref{eq:logint} with nodes $t_1,\ldots,t_m \in [0,1]$ and weights $w_1,\ldots,w_m > 0$ to get a function
\begin{equation}
r(x) = \sum_{i=1}^m w_i f(t_i,x).
\end{equation}
For example, in \cite{log_approx} the authors used Gaussian quadrature to choose the weights and nodes of $r(x)$. The resulting function agreed with the first $2m+1$ derivatives (i.e., derivatives $0,\ldots,2m$) of the logarithm function. However, this approximation to the logarithm is unsuitable for the current work as it is not a global lower bound on $\ln$.\footnote{In fact, it is neither a lower bound nor an upper bound.}

It turns out though that if we use Gauss-Radau quadrature then we get rational functions that are global bounds on $\ln$. Gauss-Radau quadrature~\cite[p.103]{davis1984methods}
is a variant of Gaussian quadrature where one of the endpoints of the integral interval is required to be a node of the quadrature formula.
\begin{theorem}[Gauss-Radau quadrature]
	\label{thm:gauss-radau-quad}
	For any integer $m \geq 1$, there exist nodes $t_1,\ldots,t_{m-1} \in (0,1)$ and weights $w_1,\ldots,w_m > 0$ such that the quadrature formula
	\begin{equation}
		\label{eq:gaussradau1}
		\int_{0}^{1} g(t) dt \approx \sum_{i=1}^{m-1} w_i g(t_i) + w_m g(1)
	\end{equation}
	holds with equality for all polynomials $g$ of degree up to $2m-2$. Furthermore $w_m = 1/m^2$.
\end{theorem}
We note that such nodes and weights can be expressed in terms of properties of Legendre polynomials~\cite[p.103]{davis1984methods} and can be numerically computed efficiently~\cite{golub1973some}. 
If we use this quadrature formula in Theorem~\ref{thm:gauss-radau-quad} for the integral representation of the logarithm \eqref{eq:logint} we get the following rational approximation of $\ln$:
\begin{align}
	\label{eq:def_rm}
	r_m(x) := \sum_{i=1}^{m-1} w_i f(t_i,x) + w_m f(1,x) = \sum_{i=1}^{m-1} w_i \frac{x - 1}{t_i(x-1) + 1} + w_m \frac{x-1}{x} . 
\end{align}
The next proposition shows that the rational functions $r_m$ are lower bounds on $\ln(x)$ that monotonically converge to $\ln(x)$.
\begin{proposition}
	\label{prop:rm1props}
	For any positive integer $m$, the function $r_m$ satisfies
	\begin{enumerate}
		\item $\ln(x) - r_m(x) = O((x-1)^{2m})$ for $x \rightarrow 1$,
		\item $r_m(x) \to \ln(x)$ as $m \to \infty$ for any $x > 0$,
		\item For any $x > 0$, $r_1(x) = 1-1/x \leq r_m(x) \leq r_{m+1}(x) \leq \ln(x)$
	\end{enumerate}
\end{proposition}
\begin{proof}
	See Appendix \ref{sec:proofsrm}.
\end{proof}

\begin{remark}
\label{rem:connectionpade}
\begin{itemize}
\item One can obtain quantitative bounds on the approximation $\ln(x) - r_m(x)$ showing that the convergence is uniform on any compact segment $[\epsilon,\epsilon^{-1}] \subset (0,\infty)$ with a convergence rate $\approx (1-\sqrt{\epsilon})^{2m}$. This can be proved using similar tools as \cite[Proposition 6]{log_approx}. { Furthermore, one can also show relative approximation error bounds that hold for all $x > 0$, such as
\begin{equation}
0 \leq \ln(x) - r_m(x) \leq \frac{1}{m^2} \frac{(x-1)^2}{x} \quad \forall x > 0.
\end{equation}
See \cite[Proposition 2.2]{squashedsdp} and \cite{oisinpaper} for more details.}
\item Since $\ln$ satisfies $\ln(x^{-1}) = -\ln(x)$ it is natural to ask what is $r_m(x^{-1})$? Note that the rational functions $f(t,x)$ satisfy $f(t,x^{-1}) = -f(1-t,x)$. As such one can verify that $r_m$ satisfies
\begin{equation}
\label{eq:rmrmbar}
r_m(x^{-1}) = -\bar{r}_m(x)
\end{equation}
where $\bar{r}_m(x)$ is the rational function obtained by applying an $m$-point Gauss-Radau quadrature on \eqref{eq:logint} with an endpoint at 0 (instead of 1). The resulting function $\bar{r}_m(x)$ is an \emph{upper bound} on $\ln(x)$ and also satisfies $\bar{r}_m(x) - \ln(x) = O((x-1)^{2m})$ for $x\to 1$. Noting that $f(0,x) = x-1$ we see that $\bar{r}_m(x)$ is a $(m,m-1)$ rational function\footnote{I.e., it can be written as the ratio of two polynomials $F/G$ where $\deg F = m$ and $\deg G = m-1$}, and as such it is the corresponding Pad\'e approximant of $\ln$ at $x=1$. Many properties of $\bar{r}_m(x)$ are discussed in detail in \cite{topsoe2007some}.
\item { One can obtain similar rational approximations for Petz divergences. The function $x\mapsto 1-x^{-\alpha}$ is operator monotone for $\alpha \in (0,1)$, and has the following integral representation
\begin{equation}
\label{eq:xalpha}
1-x^{-\alpha} = \int_{0}^{1} \frac{x-1}{1+t(x-1)} d\mu_{\alpha}(t)
\end{equation}
where $d\mu_{\alpha}(t) = \frac{\sin(\alpha \pi)}{\pi} t^{\alpha} (1-t)^{-\alpha}$. To obtain rational approximations for $1-x^{-\alpha}$, we can also apply the Gauss-Radau quadrature for the integral \eqref{eq:xalpha}, this time with respect to the measure $d\mu_{\alpha}(t)$ (see e.g., \cite{gautschi2000gauss}). This yields nodes $\{t_{\alpha,i}\}$ and weights $\{w_{\alpha,i}\}$ for $i=1,\ldots,m$, with $t_{\alpha,m}=1$ that can be used to define the rational function
\begin{equation}
r_{\alpha,m}(x) = \sum_{i=1}^m w_{\alpha,i} f(t_{\alpha,i},x).
\end{equation}
In \cite{oisinpaper} it is shown that the functions $r_{\alpha,m}(x)$ satisfy $r_{\alpha,m}(x) \leq r_{\alpha,m+1}(x) \leq 1-x^{-\alpha}$ for all $x > 0$. Furthermore, if one uses the Gauss-Radau quadrature with an endpoint fixed at $t=0$ (instead of $t=1$), then the resulting rational functions are upper bounds on $1-x^{-\alpha}$. We refer to \cite{oisinpaper} for further properties about these rational functions, including bounds on the approximation error as a function of $m$.
}
\end{itemize}
\end{remark}

\subsection{Variational expressions approximating the quantum relative entropy}

We now state and prove the main result of this work, a sequence of converging variational upper bounds on the relative entropy.

\begin{theorem}\label{thm:D_upper_bounds}
	Let $\rho, \sigma$ be two positive semidefinite linear functionals on a von Neumann algebra $\cA$ such that $Q_2(\rho\|\sigma) < \infty$ (see Definition~\ref{def:alpha-quasi}).
	Then for any $m \in \NN$ and the choice of $t_1,\dots t_m \in (0,1]$ and $w_1,\dots,w_m > 0$ as in Theorem~\ref{thm:gauss-radau-quad}, we have
	\begin{equation}
	\label{eq:DubvN}
		D(\rho\|\sigma) \leq 
		-\sum_{i=1}^m \frac{w_i}{t_i \ln 2}\,\inf_{a \in \cA} \left\{\rho(I) + \rho(a+a^*) + (1-t_i) \rho(a^* a) + t_i\sigma(aa^*) \right\} \ .
	\end{equation}
	Moreover, the RHS converges to $D(\rho \|\sigma)$ as $m \to \infty$.
	
	In the special case where $\cA = B(H)$ for a separable Hilbert space $H$ and $\rho$ and $\sigma$ are defined via trace-class operators on $H$ (also denoted by $\rho$ and $\sigma$) satisfying $\rho \leq \lambda \sigma$ for some $\lambda \in \RR_+$, we can give an explicit bound on the norm of the operators appearing in the optimization: 
	\begin{equation}\label{eq:D_ub_boundedops_vN}
		D(\rho\|\sigma) \leq 
		-c_m -\sum_{i=1}^{m-1} \frac{w_i}{t_i \ln 2}\,\inf_{Z \in B(H), \| Z \| \leq \alpha_i} \left\{ \tr{\rho (Z+Z^*)} + (1-t_i) \tr{\rho Z^* Z} + t_i \tr{\sigma Z Z^*} \right\} \ ,
	\end{equation}
	where $c_m = \tr{\rho}\left( \sum_{i=1}^m \frac{w_i }{t_i \ln 2} - \frac{\lambda}{m^2\ln 2}\right)$ and $\alpha_i = \frac{3}{2} \max\{\frac{1}{t_i}, \frac{\lambda}{1-t_i}\}$.
	Moreover, the RHS converges to $D(\rho \|\sigma)$ as $m \to \infty$.
	\end{theorem}
	
\begin{remark}
This result can be seen as computationally tractable approximations of the variational expression for the quantum relative entropy due to Kosaki~\cite{Kos86}, a variant of which is also present in the work of Donald~\cite{Don86}; we refer to the book~\cite{OP04} for more information on these expressions.  For example in Donald~\cite[Section 4]{Don86} it is shown that
\begin{equation}
D(\rho \| \sigma) = -\inf_{\substack{0 < \delta < \epsilon < 1\\ a(t),b(t) \in \cA\\ a(t)+b(t)=I}} \int_{\delta}^{\epsilon} \rho\left(\frac{1}{t} a(t)a(t)^* - b(t)b(t)^*\right) + \sigma(b(t)^* b(t)) dt,
\end{equation}
where the functions $a(t),b(t):(\delta,\epsilon)\to \cA$ are piecewise constant. Importantly however, each one of our approximations is guaranteed to be an upper bound on the relative entropy.
\end{remark}

	\begin{proof}
		The relative entropy $D$ is defined via the quasi-relative entropy $D_F$ for the function $F(x,y) = y \log_2 y - y \log_2 x$. Using the property of $\nu_{\rho, \sigma}$~\eqref{eq:def_div_int}, we have
		\begin{equation}
		\begin{aligned}
			D(\rho \| \sigma) 
			&= - \frac{1}{\ln 2} \int_{\RR_+^2} y \ln(x/y) \mathrm{d}\nu_{\rho, \sigma}(x,y) \ .
		\end{aligned}
		\end{equation}
		Using Proposition~\ref{prop:rm1props}, we can lower bound $\ln(x/y) \geq r_m(x/y)$ (defined in~\eqref{eq:def_rm} and nodes and weights in~\eqref{thm:gauss-radau-quad}) obtaining
		\begin{equation}
		\begin{aligned}
			D(\rho \|\sigma) &\leq -  \sum_{i=1}^m \frac{w_i}{\ln{2}} \int_{\RR_+^2} y \frac{x/y - 1}{t_i (x/y - 1) + 1}  \mathrm{d}\nu_{\rho, \sigma} (x,y) \\ 
			&=  -\sum_{i=1}^m \frac{w_i}{\ln{2}} D_{F_{t_i}}(\rho\|\sigma)
		\end{aligned}
		\end{equation}
		where on the second line we have defined the quasi-relative entropy $D_{F_{t_i}}(\rho\|\sigma)$ with $F_{t_i}(x,y) := y\frac{x-y}{t_i(x-y)+y}$. Now applying Proposition~\ref{prop:varparallelsum} to each $D_{F_{t_i}}$ we get 
		\begin{equation}
		\begin{aligned}
			D(\rho\|\sigma) &\leq - \sum_{i=1}^m \frac{w_i}{t_i \ln 2} \inf_{a \in \cA} \left\{\rho(I) + \rho(a+a^*) + (1-t_i) \rho(a^* a) + t_i\sigma(aa^*)\right\} \ .
		\end{aligned}
		\end{equation}
		This concludes the proof of the upper bounds. 
		
		We now turn to the proof of convergence. 
		Let $R_m(x,y) := \sum_{i=1}^m w_i F_{t_i}(x,y)$. Then we have already shown that $D(\rho\|\sigma) \leq \sum_i -\frac{w_i}{\ln 2} D_{F_{t_i}}(\rho \| \sigma)  = \frac{1}{\ln 2} D_{-R_m}(\rho\|\sigma)$. Furthermore, we have
		\begin{equation}
		\label{eq:diffDrmDumegaki}
		\frac{1}{\ln 2} D_{-R_m}(\rho \| \sigma) - D(\rho \| \sigma) = \frac{1}{\ln 2}\int_{\RR_+^2} y \left( \ln(x/y) - r_m(x/y) \right) d\nu_{\rho, \sigma}(x,y) \ .
		\end{equation}
		For any fixed $(x,y) \in \RR^2_+$ we know from Proposition~\ref{prop:rm1props} that $y \left( \ln(x/y) - r_m(x/y) \right)$ converges monotonically to 0 as $m\to \infty$. Thus, by the monotone convergence theorem we get that $\frac{1}{\ln 2} D_{-R_m}(\rho \| \sigma) - D(\rho \| \sigma) \to 0$ as desired.
		
		To obtain the second statement with explicit bounds on the operators, we use the second part of Proposition~\ref{prop:varparallelsum}. It suffices to use the variational expression in~\eqref{eq:var_expr_bound_norm} for the terms $i \in \{1,2,\dots,m-1\}$ instead of the general form. We get:
		\begin{equation}
		\begin{aligned}
			D(\rho\|\sigma) &\leq - \sum_{i=1}^{m-1} \frac{w_i}{t_i \ln 2} \inf_{\substack{Z \in B(H) \\ \| Z \| \leq \alpha_i}} \left\{\tr{\rho} + \tr{\rho(Z+Z^*)} + (1-t_i) \tr{\rho Z^* Z} + t_i \tr{\sigma Z Z^*} \right\} - \frac{w_m}{\ln 2} D_{F_1}( \rho \| \sigma) \ .
		\end{aligned}
		\end{equation}
		It now only remains to find a lower bound on $D_{F_1}( \rho \| \sigma)$ using the property $\rho \leq \lambda \sigma$. Note that
		\begin{equation}
			\begin{aligned}
			D_{F_1}( \rho \| \sigma) &= \inf_{Z \in B(H)} \left\{ \tr{\rho} + \tr{\rho(Z+Z^*)} + \tr{\sigma ZZ^*} \right\} \\
			&\geq \inf_{Z \in B(H)} \left\{ \tr{\rho} + \tr{\rho(Z+Z^*)} + \frac{1}{\lambda} \tr{\rho ZZ^*} \right\} \\
			&= \inf_{Z \in B(H)} \left\{ \tr{\rho} - \lambda \tr{\rho} + \tr{\rho (\sqrt{\lambda} I + \frac{1}{\sqrt{\lambda}} Z ) (\sqrt{\lambda} I + \frac{1}{\sqrt{\lambda}} Z )^*} \right\} \\
			&= \tr{\rho} - \lambda \tr{\rho} \ ,
			\end{aligned}
		\end{equation} 
		which together with the fact that $w_m = \frac{1}{m^2}$ leads to the desired bound.
	\end{proof}

\section{Conclusion}

In this work we derived a converging sequence of upper bounds on the relative entropy between two positive semidefinite linear functionals on a von Neumann algebra. We then demonstrated how to use this sequence of upper bounds to derive a sequence of lower bounds on the conditional von Neumann entropy. The resulting optimization problems could then be relaxed to a convergent hierarchy of semidefinite programs using the NPA hierarchy. Overall this gives a computationally tractable method to compute lower bounds on the rates of DI-RE and DI-QKD protocols. 

We applied our method to compute lower bounds on the asymptotic rates of various device-independent protocols. We compared the rates derived with our technique to other numerical techniques~\cite{TSGPL19, BFF21, BRC21} and to known tight analytical bounds~\cite{PABGMS, GMKB21b, WAP20}. We found everywhere substantial improvements over the previous numerical techniques and we also demonstrated that our technique could recover all known tight analytical bounds, making it the first general numerical technique to do so and showing that the method can converge quickly. Compared with the previous numerical techniques, not only does our technique derive higher rates but it can also be much faster due to the fact that the noncommutative polynomial optimization that we derive is of low degree and can be ran without additional operator inequalities  without affecting the rates in a substantial way. Furthermore, our derivation of this optimization problem was done in the general setting of bounded operators on a separable Hilbert space, allowing it to compute the rates of device-independent protocols even when the systems are infinite dimensional. This is in contrast to the previous general numerical techniques of~\cite{TSGPL19, BFF21} that assumed finite dimensional Hilbert spaces in their analyses.

Computing key rates for DI-QKD, we also found significant improvements on the minimal detection efficiency required to generate secret key using a pair of entangled qubits. In particular we found a minimal detection efficiency lower than $0.8$ which is now well within the regime of current device-independent experiments~\cite{diexperiment1,diexperiment2,diqkd_exper1,diqkd_exper2,diqkd_exper3}. With further effort we can hope to increase the rates of these experiments further and we leave a thorough investigation of this question to future work.

We also demonstrated that min-tradeoff functions, necessary for the entropy accumulation theorem~\cite{DFR,DF}, could be derived directly from our numerical computations. Therefore our technique can also be used to compute the rates of actual finite round protocols and to subsequently prove their security. For example, this then opens the possibility of using the global randomness bounds (see Figure~\ref{fig:chsh_global_comparison}) for the CHSH game to improve the rates of the recent DI-RE experiments~\cite{diexperiment1,diexperiment2}.

There remain a few open questions from this work. 
It would be interesting to investigate how we could make our computations more efficient. In particular, when exploring more complex scenarios the complexity of the SDPs can quickly grow (see the difference in runtimes between Fig.~\ref{fig:chsh_comparison} and Fig.~\ref{fig:chsh_global_comparison}). To combat this one could search for other variational forms that converge faster to the von Neumann entropy. Otherwise, what are the best monomial sets to include in our relaxations or can we exploit symmetries to reduce the size of the SDPs~\cite{R18}? When computing the NPA relaxations we made various simplifications to the problem (see Remark~\ref{rem:faster_computations} and the captions of the various figures), in what settings can these simplifications be made without losing tightness for the rate curves? In a different direction, one could also look at applying our bounds on the relative entropy to bound other entropic quantities.

\section*{Acknowledgements}
We thank Roger Colbeck and Federico Grasselli for supplying data and information that facilitated some of the comparisons in the plots. We also thank Mateus Araújo for feedback on an earlier version of the manuscript, we thank Jean-Daniel Bancal for useful discussions that led to some improved numerics and we thank Glaucia Murta for pointing out an error in a previous version.  This work is funded by the European Research Council (ERC Grant AlgoQIP, Agreement No. 851716).

\bibliographystyle{quantum}
\bibliography{tradeoff}

\appendix

\section{Proof of Proposition \ref{prop:rm1props} on the rational functions $r_m$}
\label{sec:proofsrm}

	\textit{1}. We start by proving that $\ln(x) - r_m(x) = O((x-1)^{2m})$. We can write the Taylor expansion of $f(t,x)$ (defined in~\eqref{eq:def_ftx}) about $x=1$ as \begin{equation}
	f(t,x) = (x-1)\sum_{k=0}^{\infty} (-1)^k (x-1)^k t^k,
	\end{equation}
	which is valid for $|t(x-1)| < 1$ and similarly for $\ln(x) = (x-1) + \sum_{k=1}^\infty (-1)^k \frac{(x-1)^{k+1}}{k+1}$ which is valid for $|x-1|<1$. Let $\nu_m = \sum_{i=1}^m w_i \delta_{t_i}$ be the discrete measure coming from the Gauss-Radau quadrature, then we have for all $|x-1|<1$
	\begin{equation}
		\label{eq:log-rseries}
		\begin{aligned}
			\ln(x) - r_m(x) &= (x-1) \sum_{k=0}^{\infty} (-1)^k (x-1)^k \left( \int_{0}^1 t^k dt - \int_{0}^1 t^k d\nu_m(t) \right) \\
			&= (x-1)\sum_{k=2m-1}^{\infty} (-1)^k (x-1)^k \left( \int_{0}^1 t^k dt - \int_{0}^1 t^k d\nu_m(t) \right) \\
			&= (x-1)\sum_{k=0}^{\infty} (-1)^{k + 2m-1} (x-1)^{k+2m-1} \left( \int_{0}^1 t^{k+2m-1} dt - \int_{0}^1 t^{k+2m-1} d\nu_m(t) \right) \\
			&= (x-1)^{2m}\sum_{k=0}^{\infty} (-1)^{k-1} (x-1)^{k} \left( \int_{0}^1 t^{k+2m-1} dt - \int_{0}^1 t^{k+2m-1} d\nu_m(t) \right) \\
		\end{aligned}
	\end{equation}
	where on the second line we used the fact that the Gauss-Radau formula \eqref{eq:gaussradau1} is exact for all polynomials of degree up to $2m-2$ means that the terms $k=0,\ldots,2m-2$ in the sum are equal to 0. Thus this gives $\ln(x) - r_m(x) = O((x-1)^{2m})$.
	
	\textit{2}. For any fixed $x > 0$ the function $t\mapsto f(t,x)$ is continuous on $[0,1]$ and so can be approximated arbitrarily closely over $[0,1]$ by polynomials. The claim then follows from the fact that the $m$-point quadrature rule is exact for all polynomials of degree $2m-2$.
	
	\textit{3}. This item can be shown using the fact mentioned in Remark \ref{rem:connectionpade} that $r_m(x) = -\bar{r}_m(x^{-1})$ where $\bar{r}_m$ is the $(m,m-1)$ Pad\'e approximant of $\ln$ at $x=1$. Indeed it was shown in \cite[Equation (27)]{topsoe2007some} that the functions $\bar{r}_m$ satisfy $\bar{r}_m(x) - \bar{r}_{m+1}(x) = \frac{(x-1)^{2m}}{S(x)}$ where $S(x)$ is a polynomial positive on $(0,\infty)$. This shows that $\bar{r}_m(x) \geq \bar{r}_{m+1}(x)$ and as such $r_m(x) \leq r_{m+1}(x)$. Since $\lim_{m\to \infty} r_m(x) = \ln(x)$, this also shows $r_m(x) \leq \ln(x)$.

\section{Proof of Lemma \ref{lem:lyapeq} concerning the Sylvester equation \eqref{eq:equation_Z}}
\label{sec:lyapeq}

	We use a result of \cite{pedersen1976operator} to show the existence of a bounded operator $Z$ satisfying~\eqref{eq:equation_Z}. Given the positive semidefinite operators $\rho$ and $\sigma$ on $H$, we construct the operators on the direct sum $H \oplus H$: $X = \begin{pmatrix} (1-t) \rho & 0 \\ 0  & t \sigma \end{pmatrix}$ and $K = \frac{1}{2} \begin{pmatrix} \rho & \rho \\ \rho & \lambda \sigma \end{pmatrix}$. Observe first that both $X$ and $K$ are positive semidefinite operators. In addition, recalling that $\alpha=\max\{ \frac{3}{2(1-t)}, \frac{3\lambda}{2t} \}$, we get
	\begin{equation}
		\begin{aligned}
			\alpha X - K = \begin{pmatrix} (\alpha (1-t) - \frac{1}{2}) \rho & -\rho \\ -\rho & (\alpha t - \frac{\lambda}{2}) \sigma \end{pmatrix} \geq \begin{pmatrix} \rho & -\rho \\ -\rho & \rho \end{pmatrix}  \geq 0 \ .
		\end{aligned}
	\end{equation}
	We are now ready to apply \cite[Theorem 3.1]{pedersen1976operator}. Note that using the fact that $K \leq \alpha X$, the operator monotonicity of the square root function, and the inequality $(a^2 + ax)^{1/2} \leq a + \frac{x}{2}$ for $a,x \geq 0$, we get for any $s > 0$
	\begin{equation}
		\begin{aligned}
			(X^2 + 2s K)^{1/2} - X &\leq (X^2 + 2 \alpha s X)^{1/2} - X\\
			 &= \begin{pmatrix} ((1-t)^2 \rho^2 + 2\alpha s (1-t) \rho)^{1/2} - (1-t) \rho & 0 \\ 0 & (t^2 \sigma^2 + 2\alpha s t \sigma)^{1/2} - t \sigma  \end{pmatrix} \\
		&\leq \begin{pmatrix} \alpha s I & 0 \\ 0 & \alpha s I  \end{pmatrix} \\
		&\leq \alpha s I \ .
		\end{aligned}
	\end{equation}
	Thus we can apply~\cite[Theorem 3.1]{pedersen1976operator} and obtain the existence of a positive semidefinite operator $T$ on $H \oplus H$ with operator norm bounded by $\alpha$ that satisfies:
	\begin{align}
		\label{eq:xtk}
		XT + TX = 2K \ .
	\end{align}
	Writing $T = \begin{pmatrix} T_{11} & Z \\ Z^* & T_{22} \end{pmatrix}$, the equality~\eqref{eq:xtk} implies 
	\begin{align}
		\label{eq:Z-satisfies}
		(1-t) \rho Z + t Z \sigma = \rho
	\end{align}
and we have $\|Z\| \leq \|T\| \leq \alpha$.

\section{The NPA hierarchy}\label{app:npa}
In this section we briefly detail the NPA hierarchy which allows us to relax non-commutative polynomial optimization problems to a hierarchy of semidefinite programs. For a full exposition, we refer the reader to the original work~\cite{NPA_general}.

We begin by recalling the basic definitions of noncommutative polynomial optimization problems. We consider noncommutative polynomials in the variables $x=(x_1,\ldots,x_n)$ and their adjoints $x^*=(x_1^*,\ldots,x_n^*)$. A \emph{word} corresponds to a monomial constructed from the variables above, i.e. products of the variables. We let $\cA_{k}$ be the set of noncommutative polynomials that are linear combinations of words with $|w| \leq k$. Similarly let $\cW_{k}$ be the set of words with $|w| \leq k$. We are interested in a noncommutative polynomial optimization problem of the form
\begin{equation}
	\label{eq:ncopt-constrained}
	\begin{array}{ll}
		c_{\mathrm{opt}} = \displaystyle\inf_{\substack{\cH, \psi \in \cH, \|\psi\|^2=1\\ X_1,\ldots,X_n \in \cB(\cH)}} & \< \psi , p(X) \psi\>\\
		\qquad\qquad\text{subject to} & \langle \psi, r_i(X) \psi \rangle \geq b_i \\
		& q_j(X) \geq 0
	\end{array}
\end{equation}

Let $Q = \{q_j : j =1,\dots, m\}$ be the set of positive polynomials in \eqref{eq:ncopt-constrained}. The \emph{quadratic module} $M_Q$ is the set of all polynomials of the form  $\sum_{ij} f_i^* f_i + \sum_{ij} g_{ij}^* q_j g_{ij}$. If $M_Q$ contains the polynomial $C - \sum_{i=1}^n x_i x_i^* + x_i^* x_i$ for some $C> 0$ then we say the problem is Archimedean. Note that if each of the variables $x_i$ in the problem has an explicit bound on their operator norm, i.e. $x_i^* x_i \leq C_i$  and $x_i x_i^* \leq C_i$ then the problem becomes Archimedean if these constraints are added to the problem. 

A moment relaxation of level $k$ is defined by a positive semidefinite linear functional $L: \cA_{2k} \rightarrow \CC$, i.e. for $f\in \cA_{2k}$ with $f = \sum_{w \in \cW_{2k}} f_w w$ for some $f_w \in \CC$ then $L(f) = \sum_{w\in \cW_{2k,2k}} f_w L(w)$. These positive semidefinite linear functionals are in one-to-one correspondence with so-called \emph{moment matrices of level $k$} $M_k$ which is a positive semidefinite matrix whose rows and columns are indexed by words in $\cW_{k}$ and whose $(v,w)$ entry is given by
\begin{equation}
M_k(v,w) = L(v^* w).
\end{equation}
For each $q \in Q$ we also define the \emph{localizing moment matrix} $M_{k-d_q}^{q}$ where $d_q = \lceil\deg(q)/2\rceil$ as the matrix whose rows and columns are indexed by words in $\cW_{k-d_q}$ and whose $(u,v)$ entry corresponds to
\begin{equation}
M_{k-d_q}^{q}(v,w) = L(v^{*} q w).
\end{equation}
We may then define the $k^{\mathrm{th}}$ level relaxation of \eqref{eq:ncopt-constrained} as
\begin{equation}
	\label{eq:ncopt-krelaxed}
	\begin{aligned}
		c_k = \min& \quad \sum_w p_w L(w) \\
		\mathrm{s.t.}& \quad M_{k}(1,1) = 1 \\ 
		& \quad \sum_w r_{i,w} L(w) \geq b_i \\ 
		& \quad M_{k} \geq 0 \\
		& \quad M_{k-d_q}^{q} \geq 0 \qquad \text{for all }q\in Q
	\end{aligned}
\end{equation}

If the problem is Archimedean then the authors of~\cite{NPA_general} showed that $\lim_{k \to \infty} c_k = c_{\mathrm{opt}}$.

\section{Additional plots}\label{app:additional_plots}

In this section we provide some additional plots that demonstrate the convergence of our technique. In particular we recover instances of other known tight analytical bounds~\cite{GMKB21b,WAP20} in addition to~\cite{PABGMS} which was shown in Figure~\ref{fig:chsh_comparison}.

In Figure~\ref{fig:holz_comparison} we demonstrate that our technique can be used to recover the tight analytical bound~\cite{GMKB21b} on $\inf H(A|X=0,Q_E)$ for devices constrained to violate the Holz inequality~\cite{HKB20} for three parties. Where the Holz inequality is formulated as follows: let Alice, Bob and Charlie each have binary input devices that give outputs in $\{1, -1\}$. Given a projective measurement $\{A_{1|x}, A_{-1,x}\}$ for Alice on input $x$ let $A_x = A_{1|x} - A_{-1|x}$ be the corresponding observable. Define observables $B_y$ and $C_z$ for Bob and Charlie in the same way. Finally, given a pair of observables $\{X_0, X_1\}$ define $X_{\pm} := \frac12 (X_0 \pm X_1)$. Then the Holz Bell-expression for three parties is given as 
\begin{equation}
 \cB_H	= \langle A_1 B_+ C_+ \rangle - \langle A_0 (B_- + C_-) \rangle - \langle B_- C_- \rangle,
\end{equation}
it has a classical bound of $1$ and a quantum mechanical bound of $3/2$. In the figure we see that, in a similar fashion to Figure~\ref{fig:chsh_comparison}, we recover the tight analytical bound at a low NPA relaxation level and rapidly in the size of Gauss-Radau quadrature.
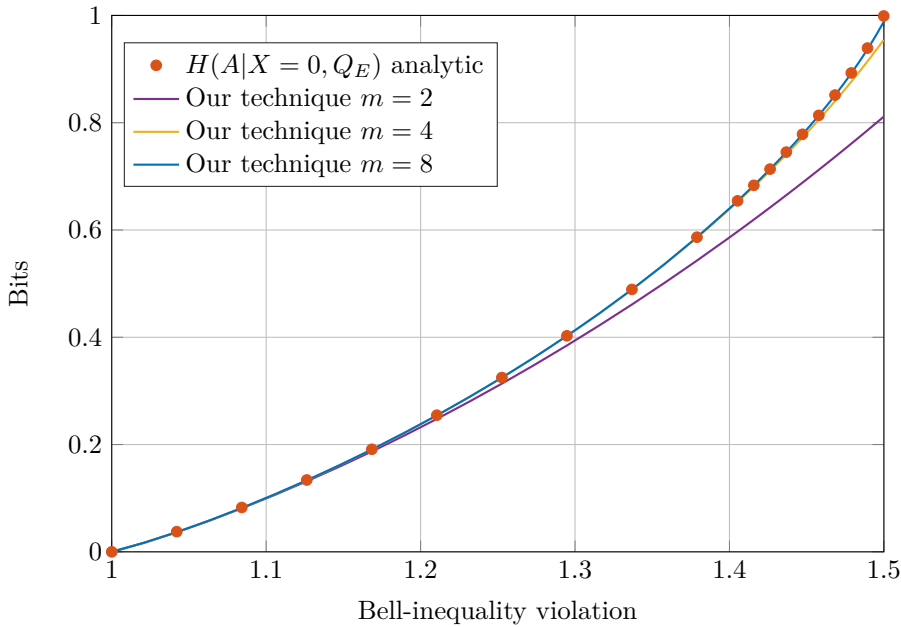
\begin{figure}
	\centering
	\definecolor{mycolor2}{rgb}{0.00000,0.44700,0.74100}%
	\definecolor{mycolor4}{rgb}{0.85000,0.32500,0.09800}%
	\definecolor{mycolor3}{rgb}{0.92900,0.69400,0.12500}%
	\definecolor{mycolor1}{rgb}{0.49400,0.18400,0.55600}%
	\begin{tikzpicture}
		
		\begin{axis}[%
			width=4in,
			height=2.8in,
			scale only axis,
			xmin=1.0,
			xmax=1.5,
			ymin=0,
			ymax=1,
			grid=major,
			xlabel={Bell-inequality violation},
			ylabel={Bits},
			xtick={1.0, 1.1, 1.2, 1.3, 1.4, 1.5},
			axis background/.style={fill=white},
			legend style={at={(0.5,0.95)},legend cell align=left, align=left, draw=white!15!black}
			]
			\addplot[color=mycolor4, line width=0.0pt, mark=*, mark repeat = 2, only marks] table[col sep=comma] {analytic_H_holz_local_tikz.dat};
			\addlegendentry{$H(A|X=0,Q_E)$ analytic}	
			\addplot[color=mycolor1, line width=0.8pt] table[col sep=comma] {kosaki_holz_local_2M_tikz.dat};
			\addlegendentry{Our technique $m=2$}
			\addplot[smooth, color=mycolor3, line width=0.8pt] table[col sep=comma] {kosaki_holz_local_4M_tikz.dat};
			\addlegendentry{Our technique $m = 4$}
			\addplot[color=mycolor2, line width=0.8pt] table[col sep=comma] {kosaki_holz_local_8M_tikz.dat};
			\addlegendentry{Our technique $m = 8$}
		\end{axis}
	\end{tikzpicture}%
	\caption{\textbf{Recovering the tight analytical bound of~\cite{GMKB21b}.} Comparison of lower bounds on $H(A|X=0,Q_E)$ for quantum devices that constrained to achieve some minimal violation of the Holz inequality for three parties~\cite{HKB20}. Numerical bounds were computed at a relaxation level $2 + ABC$. A single SDP takes a few seconds to run at this level.}
	
	\label{fig:holz_comparison}
\end{figure} 

In Figure~\ref{fig:asymchsh_comparison} we look at recovering instances of the family of tight analytical bounds on the asymmetric CHSH inequality~\cite{WAP20}. The asymmetric CHSH inequality generalizes the standard CHSH inequality in the following way. Let the observables for Alice and Bob be defined in the same way as above, then the asymmetric CHSH expression with weight $\alpha \in \RR$ is given as 
\begin{equation}
	\cB_{\alpha} = \alpha( \langle A_0 B_0 \rangle + \langle A_0 B_1 \rangle) + \langle A_1 B_0 \rangle - \langle A_1 B_1 \rangle.
\end{equation}
We see that for $\alpha = 1$ we recover the standard CHSH expression. The maximum classical value is given by $\max\{ 2, 2|\alpha|\}$ and the maximum quantum value is given by $2 \sqrt{1+\alpha^2}$. In the figure we plot two instances of this family, $\alpha = 1.1$ and $\alpha = 0.9$. Similar to the Holz inequality we find that for both values of $\alpha$ we chose, we recover the analytical bound. However, note that in the case where $\alpha = 0.9$ it was slightly more challenging. If we use speedup (3) from Remark~\ref{rem:faster_computations} then for the chosen relaxation level and Gauss-Radau quadratures, we do not recover part of the analytical bound. In particular, with the splitting of the objective function it appears to be unable to recover the tight rate curve. In particular, the construction of the analytical rate curve for $|\alpha| < 1$ in~\cite{WAP20} is split into two parts, a convex function $g(s)$ is first defined which gives the rate curve for the larger violations $s$. Then the actual rate curve $r(s)$ is defined as 
\begin{equation}
	r(s) := \begin{cases}
		g(s) \qquad &\text{if } |\alpha| \geq 1 \text{ or } s \geq s_* \\
		g'(s_*) (|s| - 2) &\text{otherwise}
	\end{cases}
\end{equation}
where $s_* \in [2 \sqrt{ 1+ \alpha^2 - \alpha^4}, 2 \sqrt{1+ \alpha^2}]$ is the unique point such that the tangent of $g$ gives a value of $0$ at $s=2$. Basically, the construction is to take the function $g$ to be the rate curve until the tangent of $g$ intersects the point $(2,0)$ (moving from high scores to low scores) then the remaining rate curve is given by this tangent function. Numerical testing indicates that it is the flat part of the curve that the split objective (speedup (3) in Remark~\ref{rem:faster_computations}) cannot capture. Nevertheless, at the cost of more computational time, we can still run the numerics without speedup (3) and we're able to recover all of the rate curve in this way.
\begin{figure}
	\centering
	\definecolor{mycolor2}{rgb}{0.00000,0.44700,0.74100}%
	\definecolor{mycolor4}{rgb}{0.85000,0.32500,0.09800}%
	\definecolor{mycolor3}{rgb}{0.92900,0.69400,0.12500}%
	\definecolor{mycolor1}{rgb}{0.49400,0.18400,0.55600}%
	\centering
	\subfloat[Rates for asymmetric CHSH inequality ($\alpha = 1.1$).]{
		\begin{tikzpicture}
			
			\begin{axis}[%
				width=4in,
				height=2.8in,
				scale only axis,
				xmin=2.2,
				xmax=3,
				ymin=0,
				ymax=1,
				grid=major,
				xlabel={Bell-inequality violation},
				ylabel={Bits},
				xtick={2.2, 2.4, 2.6, 2.8, 3.0},
				axis background/.style={fill=white},
				legend style={at={(0.5,0.95)},legend cell align=left, align=left, draw=white!15!black}
				]
				\addplot[color=mycolor4, line width=0.0pt, mark=*, mark repeat = 1, only marks] table[col sep=comma] {analytic_H_asymchsh_local_alpha1.1_tikz.dat};
				\addlegendentry{$H(A|X=0,Q_E)$ analytic}	
				\addplot[color=mycolor1, line width=0.8pt] table[col sep=comma] {kosaki_asymchsh_local_alpha1.1_2M_tikz.dat};
				\addlegendentry{Our technique $m=2$}
				\addplot[smooth, color=mycolor3, line width=0.8pt] table[col sep=comma] {kosaki_asymchsh_local_alpha1.1_4M_tikz.dat};
				\addlegendentry{Our technique $m = 4$}
				\addplot[color=mycolor2, line width=0.8pt] table[col sep=comma] {kosaki_asymchsh_local_alpha1.1_8M_tikz.dat};
				\addlegendentry{Our technique $m = 8$}
			\end{axis}
		\end{tikzpicture}%
		\label{fig:test1}
	}
	~\\\
	\subfloat[Rates for asymmetric CHSH inequality ($\alpha = 0.9$).]{
		\begin{tikzpicture}
			
			\begin{axis}[%
				width=4in,
				height=2.8in,
				scale only axis,
				xmin=2,
				xmax=2.7,
				ymin=0,
				ymax=1,
				grid=major,
				xlabel={Bell-inequality violation},
				ylabel={Bits},
				xtick={2.0, 2.1, 2.2, 2.3, 2.4, 2.5,2.6,2.7},
				axis background/.style={fill=white},
				legend style={at={(0.5,0.95)},legend cell align=left, align=left, draw=white!15!black}
				]
				\addplot[color=mycolor4, line width=0.0pt, mark=*, mark repeat = 1, only marks] table[col sep=comma] {analytic_H_asymchsh_local_alpha0.9_tikz.dat};
				\addlegendentry{$H(A|X=0,Q_E)$ analytic}	
				\addplot[color=mycolor1, line width=0.8pt] table[col sep=comma] {asym_chsh_local_2M_alpha0.9_full_tikz.dat};
				\addlegendentry{Our technique $m=2$}
				\addplot[smooth, color=mycolor3, line width=0.8pt] table[col sep=comma] {asym_chsh_local_4M_alpha0.9_full_tikz.dat};
				\addlegendentry{Our technique $m = 4$}
				\addplot[color=mycolor2, line width=0.8pt] table[col sep=comma] {asym_chsh_local_8M_alpha0.9_full_tikz.dat};
				\addlegendentry{Our technique $m = 8$}	
			\end{axis}
		\end{tikzpicture}%
		\label{fig:test2}
	}
	\caption{\textbf{Recovery of the tight analytical bounds from~\cite{WAP20}.} Comparison of lower bounds on $H(A|X=0,Q_E)$ for quantum devices that constrained to achieve some minimal violation of the asymmetric CHSH inequality for different values of the weighting parameter $\alpha$. For subfigure (a) the numerical bounds were computed at a relaxation level $2 + ABZ$ including all monomials present in the objective function and implemented using speedups (1) and (3) from Remark~\ref{rem:faster_computations}. For subfigure (b) we used speedups (1) and (2) from Remark~\ref{rem:faster_computations}.}
	
	\label{fig:asymchsh_comparison}
\end{figure}
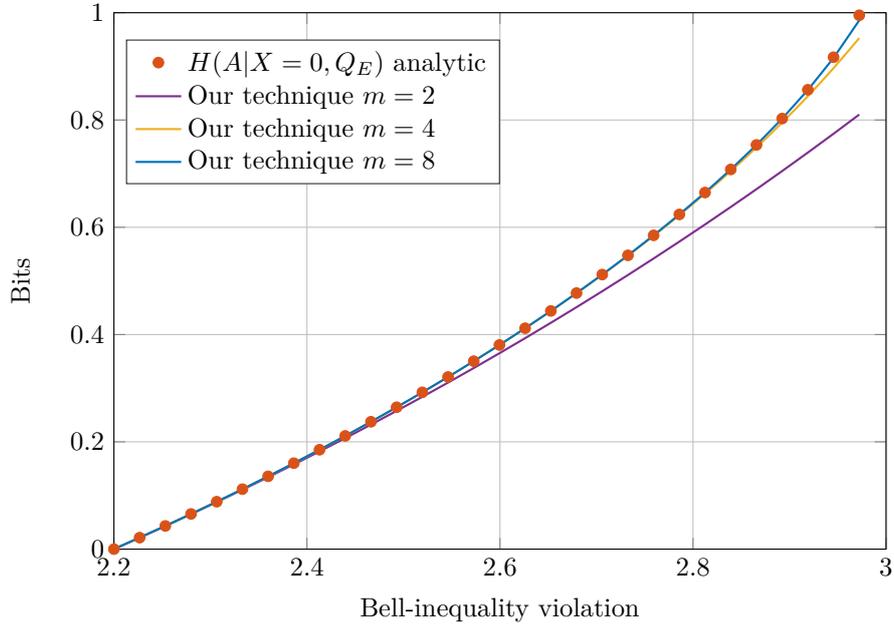
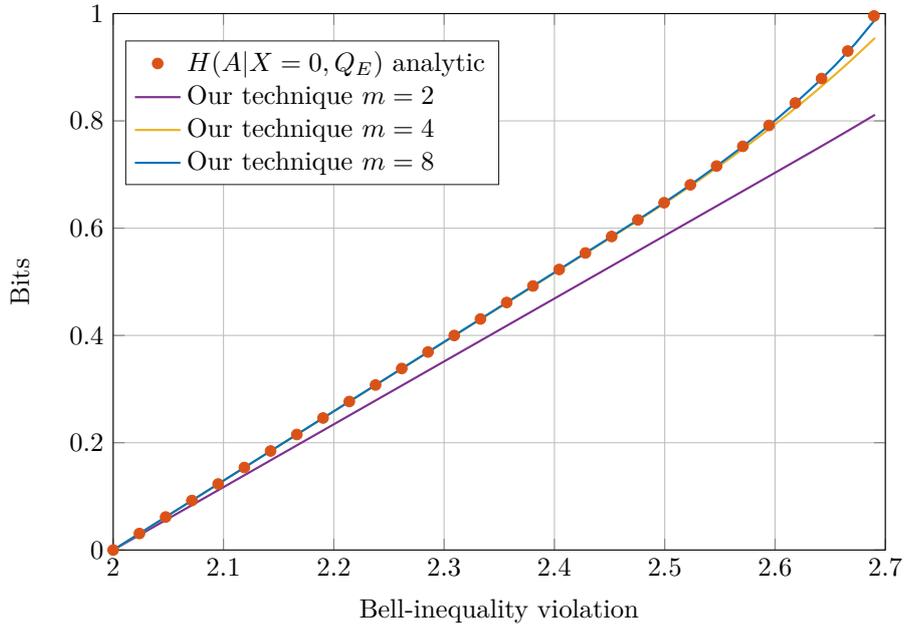

\section{Entropy accumulation and min-tradeoff functions}\label{app:eat}

In order to compute the rates of finite round DI-RE or DI-QKD protocols, one can bound the total conditional smooth min-entropy $H_{\min}^{\epsilon}(\mathbf{AB} | \mathbf{XY} E)$ accumulated by the devices during the protocol~\cite{Renner}. There are two tools, the entropy accumulation theorem~\cite{DFR,DF} and quantum probability estimation~\cite{QPE}, that allow one to break down this total smooth min-entropy into smaller, easier to compute quantities. For instance, the entropy accumulation theorem roughly states that in an $n$-round protocol 
\begin{equation}
	H_{\min}^{\epsilon}(\mathbf{AB}|\mathbf{XY}E) \geq n f(q) - O(\sqrt{n})
\end{equation}
where $f(q)$ is a so-called \emph{min-tradeoff function} defined as any function $f$ that maps probability distributions over a finite set $\mathrm{C}$ (see the discussion around~\eqref{eq:rate_optimization}) to $\RR$ such that 
\begin{equation}
	f(q) \leq \inf H(AB|XYE)
\end{equation}
where the infimum is a device-independent optimization over all strategies compatible with the average statistics $(C,q)$ observed in the protocol.

For a fixed distribution $q$ the rate calculations as performed in the main text will give a valid value for a potential min-tradeoff function. However, as we will now show, a single rate computation for some distribution $q'$ can also be used to construct a min-tradeoff function $f_{q'}$ that can then be evaluated for any distribution $q$. This construction is a consequence of weak duality for semidefinite programming problems~\cite{BoydVan}. 

Consider the following primal and dual pair of semidefinite programs
\begin{equation}
	\begin{aligned}
		p(b) = &\inf_X\quad \tr{C X} \\
		&\mathrm{s.t.} \quad \,\tr{F_i X} = b_i \qquad \qquad \text{for all } i\\
		& \qquad \,\,\, X \geq 0
	\end{aligned}
\end{equation}
and
\begin{equation}
	\begin{aligned}
		d(b) = &\sup_{\lambda,Y}\quad \sum_{i} \lambda_i b_i \\
		&\,\mathrm{s.t.} \quad \,C - \sum_i \lambda_i F_i - Y \geq 0 \\
		&\, \qquad \,\,\, Y \geq 0
	\end{aligned}
\end{equation}
where $b\in \RR^m$ is some constraint vector and $C,F_i$ are real symmetric matrices. For the purposes of this exposition the primal problem should be thought of as the NPA moment matrix relaxation, where $X$ is the moment matrix.\footnote{Note that it is sufficient to consider real moment matrices when running our computations. This is due to the fact that for any feasible complex moment matrix $\Gamma$ the matrix $(\Gamma + \overline{\Gamma})/2$ is also a feasible moment matrix with the same objective value where $\overline{\Gamma}$ denotes the matrix corresponding to elementwise complex conjugation of $\Gamma$.} The constraints $\tr{F_i X} = b_i$ can then be any general NPA constraints, as well as any statistical constraints that may have been imposed on the devices, e.g., for some statistical test $(C,q)$ we would impose constraints of the form $\sum_{(a,b,x,y): C(a,b,x,y) = c} p(a,b,x,y)  = q(c)$. We want to show that for any fixed $\hat{b}$, a feasible point $(\hat{\lambda}, \hat{Y})$ to the dual problem parameterized by $\hat{b}$ defines a function $g(b) = \sum_i \hat{\lambda}_i b_i$ that is everywhere a lower bound on the optimal primal value $p(b)$. To see this note that for any feasible point $X$ of the primal problem parameterized by $b$ we have 
\begin{equation}
\begin{aligned}
	0 &\leq \tr{X (C - \sum_i \hat{\lambda}_i F_i - \hat{Y})} \\
	&= \tr{X C} - \sum_i \hat\lambda_i \tr{X F_i} - \tr{X\hat Y} \\
	&= \tr{X C} - \sum_i \hat{\lambda}_i b_i - \tr{X\hat Y} \\
	&\leq \tr{XC} - g(b)
\end{aligned}
\end{equation}
where for both inequalities we used the fact that $\tr{AB} \geq 0$ when $A$ and $B$ are both positive semidefinite. Rearranging and taking the infimum over all feasible $X$ we find that $ g(b) \leq p(b)$. Thus, from a single solution to the dual problem we can derive a function $g$ which is everywhere a lower bound on the optimal primal value. Moreover, by modifying the choice of constraint vector $\hat{b}$ with respect to which we solve the dual we can derive different lower bounding functions.

As mentioned above in the problems we consider the constraint vector will in general consist of constraints that are fixed and constraints that correspond to the statistical test $(C,q)$ which we vary as $q$ varies. As such we can split our constraint vector $b$ into $(b_{\mathrm{fixed}}, b_{\mathrm{vary}})$ where $b_{\mathrm{vary}}$ is precisely the distribution $q$ viewed as a vector. Then by defining $\alpha = \lambda_{\mathrm{fixed}} \cdot b_{\mathrm{fixed}}$ where $\lambda_{\mathrm{fixed}}$ is the part of the dual solution $\lambda$ corresponding to the $b_{\mathrm{fixed}}$ part of $b$, we can derive a function $g(q) = \alpha + \lambda_{\mathrm{vary}} \cdot q$ where we have written $q$ in place of $b_{\mathrm{vary}}$. Thus, for any primal problem corresponding to a lower bound on rate of a protocol satisfying some statistical test $(C,q)$ , e.g.,  $\inf H(A|X=0,E)$, we immediately get an affine function $g(q)$ that satisfies 
\begin{equation}
g(q) \leq \inf H(A|X=0, E)
\end{equation}
which is in essence a min-tradeoff function. One could then use $g(q)$ together with the entropy accumulation theorem to get bounds on the rates of finite round protocols and prove security in exactly the same manner as in~\cite{BRC}.

\section{Improved bounds for noisy-preprocessing}

For this section we assume that all Hilbert spaces are finite dimensional. When deriving the statement~\eqref{eq:D_ub_boundedops} we removed the $m^\text{th}$ term in the summation $D_{F_1}$ by lower bounding it by a trivial value $\tr{\rho} - \lambda \tr{\rho}$. In the application to computing device-independent rates $\lambda = 1$ and the lower bound is $0$. This bounding was done in order to guarantee that all the operators in the subsequent NPA hierarchy relaxations were explicitly bounded which can help improve numerical stability. However, in the application to DI-QKD where we consider noisy-preprocessing we can derive a tighter bound that depends on the bitflip probability $q\in[0,1/2]$. 

First note that if $\rho$ is a state, i.e $\tr{\rho}=1$, then
\begin{equation}
\begin{aligned}
	D_{F_1}(\rho\|\sigma) = 1 - \tr{\rho^2 \sigma^{-1}}.
\end{aligned}
\end{equation}
Following the proof of Theorem~\ref{thm:D_upper_bounds} we need to lower bound this quantity or equivalently upper bound $\tr{\rho^2 \sigma^{-1}}$. In the setting of DI-QKD considered in the main text we have that after measuring but before the noisy-preprocessing Alice and Eve's joint state is
\begin{equation}
\rho_{AE} = \sum_a \outer{a} \otimes \rho_E(a)
\end{equation}
where $\rho_E(a) = \ptr{Q_AQ_B}{\rho_{Q_AQ_BQ_E}(M_{a|x^*} \otimes \id \otimes \id)}$. After Alice applies noisy-preprocessing the state transforms to 
\begin{equation}
\begin{aligned}
\widetilde{\rho}_{AE} &= \sum_{a} ((1-q)\outer{a} + q \outer{\overline a}) \otimes \rho_E(a) \\
&= \sum_a \outer{a} \otimes \left((1-q)\rho_E(a) + q\rho_E(\overline a) \right)
\end{aligned}
\end{equation}
where we have defined $\overline a = a \oplus 1$ and we can also see that $\ptr{A}{\widetilde{\rho}_{AE}} = \rho_E$, i.e. Eve's marginal state is unchanged by Alice's local data processing. We are interested in upper bounds on $\tr{\widetilde{\rho}_{AE}^2 (\id \otimes \rho_E^{-1})}$ that are smaller than the previously obtained bound of $1$. Taking the partial trace over the $A$ system we have
\begin{equation}
\tr{\widetilde{\rho}_{AE}^2 (\id \otimes \rho_E^{-1})} = \tr{\sum_a\left((1-q) \rho_E(a) + q \rho_E(\overline a) \right)^2 \rho_E^{-1}}.
\end{equation}
Now by a direct calculation we have
\begin{equation}
\begin{aligned}
	\sum_a\left((1-q) \rho_E(a) + q \rho_E(\overline a) \right)^2 &= \sum_a (1-q)^2 \rho_E(a)^2 + q^2 \rho_E(\overline a)^2 + q(1-q)(\rho_E(a)\rho_E(\overline a) + \rho_E(\overline a)\rho_E(a)) \\
	&= (1 - 2q(1-q))(\rho_E(0)^2 + \rho_E(1)^2) + 2 q (1-q)(\rho_E(0)\rho_E(1) + \rho_E(1) \rho_E(0)) \\
	&= (1 - 4q(1-q)) (\rho_E(0)^2 + \rho_E(1)^2) + 2 q(1-q) (\rho_E(0) + \rho_E(1))^2 \\
	&=(1 - 4q(1-q)) (\rho_E(0)^2 + \rho_E(1)^2) + 2 q(1-q) \rho_E^2.
\end{aligned}
\end{equation}
Applying this to the above expression we find 
\begin{equation}
\begin{aligned}
	\tr{\widetilde{\rho}_{AE}^2 (\id \otimes \rho_E^{-1})} &= (1-4q(1-q)) \tr{\rho_{AE}^2 (\id \otimes \rho_E^{-1})} + 2q(1-q) \\
	&\leq 1-2q(1-q)
\end{aligned}
\end{equation}
where we used the fact that $\tr{\rho_{AE}^2 (\id \otimes \rho_E^{-1})} \leq 1$. We can see that this bound is tight at $q=1/2$ as this implies the Petz entropy of order $2$,  $H_2(A|E) = - \log_2 \tr{\rho_{AE}^2 \rho_E^{-1}}$, must be at least $1$. This is exactly what we expect as Alice's system at $q=1/2$ is just a uniformly random bit. Similarly if $q=0$ we recover our original bound without preprocessing.

Putting everything together, we have
\begin{equation}
D_{F_1}(\rho\|\sigma) \geq 1-2 q (1-q)
\end{equation}
this implies that the constant term $c_m$ in Lemma~\ref{lem:di_rewriting} can be replaced by 
\begin{equation}
c_m = \frac{ 2 q (1-q)}{m^2 \ln 2} + \sum_{i=1}^{m-1} \frac{w_i}{t_i \ln 2} 
\end{equation}
when computing key-rates for DI-QKD protocols that include the noisy-preprocessing step.
\end{document}